\documentclass[runningheads]{llncs}
\usepackage[utf8]{inputenc}
\usepackage[T1]{fontenc}
\usepackage[english]{babel}
\usepackage{amssymb,amsmath}
\usepackage[font=small,center]{caption}
\usepackage{subcaption}
\usepackage{xcolor}
\usepackage{graphicx}
\usepackage{wrapfig}
\usepackage{makecell}
\usepackage{bussproofs}
\usepackage{tabularx}
\usepackage{makecell}
\usepackage{tikz}
\usetikzlibrary{positioning,arrows,shapes,hobby,backgrounds,calc,trees,fit,shapes.multipart}
\usepackage{tikz-cd}
\pgfdeclarelayer{background}
\pgfsetlayers{background,main}

\makeatletter
\RequirePackage[bookmarks,unicode,colorlinks=true]{hyperref}%
   \def\@citecolor{blue}%
   \def\@urlcolor{blue}%
   \def\@linkcolor{blue}%

\def\orcidID#1{\smash{\href{http://orcid.org/#1}{\protect\raisebox{-1.25pt}{\protect\includegraphics{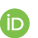}}}}}
\makeatother

\definecolor{hibou_col_lf}{RGB}{22, 22, 130}
\newcommand{\hlf}[1]{\textcolor{hibou_col_lf}{#1}}
\definecolor{hibou_col_ms}{RGB}{15, 86, 15}
\newcommand{\hms}[1]{\textcolor{hibou_col_ms}{#1}}
\newcommand{\shortColRed}[1]{\textcolor{red}{#1}}
\newcommand{\shortColBlue}[1]{\textcolor{blue}{#1}}
\newcommand{\shortColOrange}[1]{\textcolor{orange}{#1}}
\newcommand{\shortColViolet}[1]{\textcolor{violet}{#1}}

\newcommand{\globalInterleaving}{||}

\newcommand{\multiAppendLeft}{\!\!\symbol{94}\!\!}

\newcommand{\multiAppendRight}{\!\!\symbol{94}\!\!}

\newcommand{\macroElim}{\mathsf{rmv}}
\newcommand{\macroHide}{\mathsf{rmv}}

\newcommand{\macroOKVerdict}{Ok}
\newcommand{\macroKOVerdict}{Nok}

\newcommand{\macroRulePass}{\shortColBlue{R_o}}
\newcommand{\macroRuleFail}{\shortColRed{R_n}}
\newcommand{\macroRuleExec}{\shortColOrange{R_e}}
\newcommand{\macroRuleHide}{\shortColViolet{R_r}}

\newcommand{\multiDiamond}{\diamond}
\newcommand{\multiSeq}{;}
\newcommand{\multiInterleaving}{||}

\newcommand{\sigMultiPrefProjected}[1]{\overline{\sigma_{|#1}}}

\spnewtheorem*{theorem*}{Theorem}{\bfseries}{\itshape}
\spnewtheorem*{lemma*}{Lemma}{\bfseries}{\itshape}
\spnewtheorem*{property*}{Property}{\itshape}{\itshape}

\begin{document}
%
% ==========================================================
%

\title{Dealing with observability in interaction-based Offline Runtime Verification of Distributed Systems}

\titlerunning{Interaction-based Offline RV of Distributed Systems}

\author{Erwan Mahe\inst{1}\orcidID{0000-0002-5322-4337} \and Boutheina Bannour\inst{1}\orcidID{0000-0002-4943-7807} \and Christophe Gaston\inst{1}\orcidID{0000-0001-6865-5108} \and\\ Arnault Lapitre\inst{1}\orcidID{0000-0002-2185-4051} \and Pascale Le Gall\inst{2}\orcidID{0000-0002-8955-6835}  }

\authorrunning{E. Mahe, B. Bannour, C. Gaston, A. Lapitre, P. Le Gall}
% %
\institute{
Université Paris-Saclay, CEA, List, F-91120, Palaiseau, France
\and
Université Paris-Saclay, CentraleSupélec, F-91192, Gif-sur-Yvette, France
}

\maketitle

\begin{abstract}
Interactions are formal models describing asynchronous communications within a Distributed System (DS). They can be drawn in the fashion of sequence diagrams and executed thanks to an operational semantics akin to that of process algebras. Executions of DS can be characterized by tuples of local traces (one per subsystem) called multi-traces. For a given execution, those local traces can be collected via monitoring and the resulting multi-trace can be analysed using offline Runtime Verification (RV). To that end, interactions may serve as formal references. In practice, however, not all subsystems may be observed and, without synchronising the end of monitoring on different subsystems, some events may not be observed, e.g. the reception of a message may be observed but not the corresponding emission. So as to be able to consider all such cases of partial observation, we propose an offline RV algorithm which uses removal operations to restrict the reference interaction on-the-fly, disregarding the parts concerning no longer observed subsystems. We prove the correctness of the algorithm and assess the performance of an implementation.
\end{abstract}

\section{Introduction\label{sec:introduction}}

{\em Context.}
Distributed Systems (DS) have been identified in the recent survey~\cite{surveyRV_SanchezSABBCFFK19} as one of the most challenging application domains for Runtime Verification (RV). An important bottleneck is that the formal references against which system executions are analyzed are specified using formalisms or logics usually equipped with trace semantics. Indeed, because DS are composed of subsystems deployed on different computers and communicating via message passing, their executions are more naturally represented as collections of traces observed at the level of the different subsystems' interfaces rather than as single global traces \cite{constraint_based_oracles_for_timed_distributed_systems,passive_conformance_testing_of_service_choreographies}. Those collections can be gathered using a distributed observation architecture involving several local observation devices, each one dedicated to a subsystem, and deployed on the same computer as the subsystem it is dedicated to.
An approach to confront such collections of local execution traces to formal references with a trace semantics might consist in identifying the global traces that result from all possible temporal orderings of the events occurring in the local traces. If none of those global traces conforms to the formal reference, then we might conclude that an error is observed \cite{passive_conformance_testing_of_service_choreographies}. However, the absence of a global clock implies that, in all generality, it is not possible to synchronize the endings of the different local observation processes. Therefore, in the process of reconstructing global traces, some events might be missing in local traces. %Typically, the reception of a message can be observed by a local observation device, while its emission was not observed by the observation device dedicated to the subsystem that emitted it.
Such problems occur whenever, for technical or legal reasons, it is not possible to observe some subsystems or else the observation has been interrupted too early.

{\em Contributions.} In this paper, we propose a RV approach dedicated to DS with an emphasis on overcoming issues of {\em partial observability}, whether due to the absence of a global clock, or to the impossibility of observing some subsystem executions. Our approach belongs to the family of offline RV techniques in which traces are logged prior to their analysis. 
As for formal references, we inherit the framework of {\em interaction} models from earlier works \cite{revisiting_semantics_of_interactions_for_trace_validity_analysis,a_small_step_approach_to_multi_trace_checking_against_interactions}.
Interactions describe actor-oriented scenarios and can be represented graphically in the fashion of UML Sequence Diagrams (UML-SD) \cite{UML} or Message Sequence Charts (MSC) \cite{MSC}. 
In \cite{a_small_step_approach_to_multi_trace_checking_against_interactions} an algorithm to decide whether or not a collection of local traces is accepted by an interaction is given. However, this algorithm cannot cope with partial observability. The core contribution of this paper is then to define an algorithm to tackle those limitations, i.e. to deal with collections of local traces with missing or incomplete ones. Theorem \ref{th:semantics_of_hidden_interactions} will enable us to relate collections of local traces reflecting partially observed executions to those of the original reference interaction.
The key operator in our algorithm is a removal operator (Definition \ref{def:hiding}) discarding parts of the interaction relative to unobserved subsystems.
We prove the correctness of our algorithm and argue how the use of the removal operations allows us to solve partial observability (Theorem \ref{th:multipref_equates_pass}). 
Finally, we present some experiments using an implementation of our algorithm, given as an extension of the HIBOU tool~\cite{hibou_label}.

{\em Paper outline.} In Section \ref{sec:context}, we discuss the nature of DS, their modelling with interactions and the challenge of applying RV to DS. In Section \ref{sec:defs}, we define multi-traces, interactions and  associated removal operations. In Section~\ref{sec:anahide_space}, we define and prove the correctness of our RV algorithm. In Section~\ref{sec:experiments}, we report experimental results and in Section~\ref{sec:related}, we overview the related works.

\section{Preliminaries\label{sec:context}}

\paragraph{Notations} 
Given a set $A$, $A^*$ is the set of words on $A$, with $\varepsilon$ the empty word and the "$.$" concatenation law.
For any word $w \in A^*$, $|w|$ is the length of $w$ and any word $w'$ is a prefix of $w$ if there exists a word $w''$, possibly empty, such that $w = w'.w''$. Let us note $\overline{w}$ the set of prefixes of a word $w \in A^*$ and $\overline{W}$ the set of prefixes of all words of a set $W \subseteq A^*$. Given a set $A$, $|A|$ designates its cardinal and $\mathcal{P}(A)$ is the set of all subsets of $A$.

\paragraph{Distributed Systems (DS)}
From a black box perspective, the atomic concept to describe the executions of DS is that of communication {\em actions} occurring on a subsystem's interface. Here a subsystem refers to a software system deployed on a single machine. Anticipating the use of interactions as models in Section~\ref{sec:interactions}, a subsystem interface is called a \emph{lifeline} and corresponds to an interaction point on which the subsystem can receive or send some messages. Lifelines are elements of a set $\mathcal{L}$ denoting the universe of lifelines. An action occurring on a lifeline is defined by its kind (emission or reception, identified resp. by the symbols $!$ and $?$) and by the message which it carries. We introduce the universe $\mathcal{M}$ of messages.
Executions observed on a lifeline $l$ can be modelled as execution {\em traces} i.e. sequences of actions.
For $l\in \mathcal{L}$, the set $\mathbb{A}_l$ of {\em actions over $l$} is $\{ l \Delta m ~|~ \Delta \in \{!,?\}, ~ m \in {\cal M} \}$ and the set $\mathbb{T}_l$ of {\em traces over $l$} is $\mathbb{A}_l^*$. For any $a \in \mathbb{A}_l$ of the form $l?m$ or $l!m$, $\theta(a)$ refers to $l$.

Fig.\ref{fig:architecture} sketches out an example of DS composed of three remote subsystems, assimilated to their interface \hlf{\texttt{bro}}, \hlf{\texttt{pub}} and \hlf{\texttt{sub}}. This DS implements a simplified publish/subscribe scheme of communications (an alternative to client-server architecture), which is a cornerstone of some protocols used in the IoT such as MQTT \cite{website_mqtt}. The publisher \hlf{\texttt{pub}} may publish messages on the broker \hlf{\texttt{bro}} which may then forward them to the subscriber \hlf{\texttt{sub}} if it is already subscribed. 
Fig.\ref{fig:example_interaction} depicts an interaction defined between the three lifelines.
Each lifeline is depicted by a vertical line labelled by its name at the top. By default, the top to bottom direction represents time passing. That is, a communication action depicted above another one on the same lifeline occurs beforehand. 
Communication actions are represented by horizontal arrows labelled with the action's message.
Whenever an arrow exits (resp. enters) a lifeline, there is a corresponding emission (resp. reception) action at that point on the line. For example, the horizontal arrow from the lifeline \hlf{\texttt{sub}} to the lifeline \hlf{\texttt{bro}} indicates that the subsystem \hlf{\texttt{sub}} sends the message \hms{\texttt{subscribe}}, denoted as \hlf{\texttt{sub}}!\hms{\texttt{subscribe}}, which is then received by the lifeline \hlf{\texttt{bro}}, denoted as  \hlf{\texttt{bro}}?\hms{\texttt{subscribe}}. More complex behaviors can be introduced through the use of operators (similar to combined fragments in UML-SD) drawn in the shape of boxes that frame sub-behaviors of interest. For instance, in Fig.\ref{fig:example_interaction}, $loop_S$ corresponds to a sequential loop. From the perspective of the $\hlf{\texttt{bro}}$ lifeline, this implies that it can observe words of the form $(\hlf{\texttt{bro}}?\hms{\texttt{publish}})^*\hlf{\texttt{bro}}?\hms{\texttt{subscribe}}(\hlf{\texttt{bro}}?\hms{\texttt{publish}}.\hlf{\texttt{bro}}!\hms{\texttt{publish}})^*$ i.e. it can receive an arbitrary number of instances of the $\hms{\texttt{publish}}$ message then one instance of $\hms{\texttt{subscribe}}$ and then it can receive and transmit an arbitrary number of $\hms{\texttt{publish}}$. 
A representative global trace specified by the interaction in Fig.\ref{fig:example_interaction} is (see Fig.\ref{fig:architecture_global}): \\
\noindent 
\centerline{{\small $\hlf{\texttt{sub}}!\hms{\texttt{subscribe}}.\hlf{\texttt{pub}}!\hms{\texttt{publish}}.\hlf{\texttt{bro}}?\hms{\texttt{subscribe}}.\hlf{\texttt{bro}}?\hms{\texttt{publish}}.\hlf{\texttt{bro}}!\hms{\texttt{publish}}.\hlf{\texttt{sub}}?\hms{\texttt{publish}}$}}  
\noindent 
This trace illustrates that the \hlf{\texttt{pub}} and \hlf{\texttt{sub}} lifelines can send their respective messages \hms{\texttt{publish}} and \hms{\texttt{subscribe}} in any order since there are no constraints on their ordering. In contrast, the reception of a message necessarily takes place after its emission.
Since the reception of the message \hms{\texttt{subscribe}} takes place before that of the \hms{\texttt{publish}} message, this last message necessarily corresponds to the one occurring in the bottom loop.
The global trace in Fig.\ref{fig:architecture_global} is a typical example of a trace {\em accepted} by the interaction in Fig.\ref{fig:example_interaction}, as this trace completely realizes the specified behavior by: unfolding zero times the first loop; realizing the passing of the message \hms{\texttt{subscribe}} between lifelines \hlf{\texttt{sub}} and \hlf{\texttt{bro}}; unfolding one time the second loop. None of the prefixes of this accepted trace is an accepted trace.

\begin{figure}[t]%[h]
    \centering
    \begin{minipage}{.475\linewidth}
            \begin{subfigure}[t]{.9\linewidth}
                \centering
                \scalebox{.625}{\begin{tikzpicture}
\tikzstyle{subsys}=[draw,hibou_col_lf,rectangle,line width=2pt,minimum height=2cm,minimum width=2cm,inner sep=0,outer sep=0]
\tikzstyle{element_blk}=[draw,black,rectangle,line width=2pt,minimum height=2cm,minimum width=2cm,inner sep=0,outer sep=0]
\node[subsys] (A) at (-3,0) {
    \input{figures/subfig/tikz_l1}
};
\node[subsys] (B) at (0,0) {
    \input{figures/subfig/tikz_l2}
};
\node[subsys] (C) at (3,0) {
    \input{figures/subfig/tikz_l3}
};
\draw[<->,line width=2.5pt,hibou_col_ms] (A) -- (B);
\draw[<->,line width=2.5pt,hibou_col_ms] (C) -- (B);
\node[rectangle, fill=black] (Alog) at (A.south) {};
\node[rectangle, fill=black] (Blog) at (B.south) {};
\node[rectangle, fill=black] (Clog) at (C.south) {};
\node[element_blk] (reord) at (-1,-2.75) {
    \input{figures/subfig/tikz_reord}
};
\node[draw,black,rectangle,line width=2pt,align=left] (tra) at (2,-2.75) {
$\hlf{\texttt{sub}}!\hms{\texttt{subscribe}}$\\
$\hlf{\texttt{pub}}!\hms{\texttt{publish}}$\\
$\hlf{\texttt{bro}}?\hms{\texttt{subscribe}}$\\
$\hlf{\texttt{bro}}?\hms{\texttt{publish}}$\\
$\hlf{\texttt{bro}}!\hms{\texttt{publish}}$\\
$\hlf{\texttt{sub}}?\hms{\texttt{publish}}$
};
\draw[line width=2pt,black] (Alog) -- ($(Alog) + (0,-.25)$);
\draw[line width=2pt,black] ($(Alog) + (0,-.25)$) -- ($(Blog) + (0,-.25)$);
\draw[line width=2pt,black] (Blog) -- ($(Blog) + (0,-.25)$);
\draw[line width=2pt,black] (Clog) -- ($(Clog) + (0,-.25)$);
\draw[line width=2pt,black] ($(Clog) + (0,-.25)$) -- ($(Blog) + (0,-.25)$);
\draw[->,line width=2pt,black] ($(reord.north) + (0,.5)$) -- (reord.north);
\draw[->,line width=2pt,black] (reord) -- (tra);
\end{tikzpicture}}
                \caption{Global observation\label{fig:architecture_global}}
            \end{subfigure} \\\vspace*{.5cm}\\
        \begin{subfigure}[b]{.9\linewidth}
            \centering
            \scalebox{.625}{\begin{tikzpicture}
\tikzstyle{subsys}=[draw,hibou_col_lf,rectangle,line width=2pt,minimum height=.6cm,minimum width=2cm,inner sep=0,outer sep=0]
\tikzstyle{element_blk}=[draw,black,rectangle,line width=2pt,minimum height=2cm,minimum width=2cm,inner sep=0,outer sep=0]
\node[subsys] (A) at (-3,0) {
    \texttt{\hlf{pub}}
};
\node[subsys] (B) at (0,0) {
    \texttt{\hlf{bro}}
};
\node[subsys] (C) at (3,0) {
    \texttt{\hlf{sub}}
};
\draw[<->,line width=2.5pt,hibou_col_ms] (A) -- (B);
\draw[<->,line width=2.5pt,hibou_col_ms] (C) -- (B);
\node[rectangle, fill=black] (Alog) at (A.south) {};
\node[rectangle, fill=black] (Blog) at (B.south) {};
\node[rectangle, fill=black] (Clog) at (C.south) {};
\node[draw,black,rectangle,line width=2pt,align=left,below=.5cm of Alog] (Atra) {
$\hlf{\texttt{pub}}!\hms{\texttt{publish}}$
};
\node[draw,black,rectangle,line width=2pt,align=left,below=.5cm of Blog] (Btra) {
$\hlf{\texttt{bro}}?\hms{\texttt{subscribe}}$\\
$\hlf{\texttt{bro}}?\hms{\texttt{publish}}$\\
$\hlf{\texttt{bro}}!\hms{\texttt{publish}}$
};
\node[draw,black,rectangle,line width=2pt,align=left,below=.5cm of Clog] (Ctra) {
$\hlf{\texttt{sub}}!\hms{\texttt{subscribe}}$\\
$\hlf{\texttt{sub}}?\hms{\texttt{publish}}$
};
\draw[->,line width=2pt,black] (Alog) -- (Atra);
\draw[->,line width=2pt,black] (Blog) -- (Btra);
\draw[->,line width=2pt,black] (Clog) -- (Ctra);
\end{tikzpicture}}
            \caption{Complete local observation\label{fig:architecture_local}}
        \end{subfigure} 
    \end{minipage}
    \begin{minipage}{.475\linewidth}
        \begin{subfigure}[t]{.9\linewidth}
            \centering
            \includegraphics[height=3.2cm]{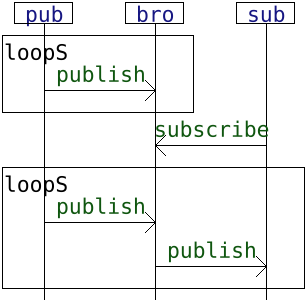}
            \caption{Interaction model\label{fig:example_interaction}}
        \end{subfigure} \\\vspace*{.5cm}\\
        \begin{subfigure}[b]{.9\linewidth}
            \centering
            \scalebox{.625}{\begin{tikzpicture}
\tikzstyle{subsys}=[draw,hibou_col_lf,rectangle,line width=2pt,minimum height=.6cm,minimum width=2cm,inner sep=0,outer sep=0]
\tikzstyle{element_blk}=[draw,black,rectangle,line width=2pt,minimum height=2cm,minimum width=2cm,inner sep=0,outer sep=0]
\node[subsys] (A) at (-3,0) {
    \texttt{\hlf{pub}}
};
\node[subsys] (B) at (0,0) {
    \texttt{\hlf{bro}}
};
\node[subsys] (C) at (3,0) {
    \texttt{\hlf{sub}}
};
\draw[<->,line width=2.5pt,hibou_col_ms] (A) -- (B);
\draw[<->,line width=2.5pt,hibou_col_ms] (C) -- (B);
\node[rectangle, fill=black] (Alog) at (A.south) {};
\node[rectangle, fill=black] (Blog) at (B.south) {};
\node[draw,black,rectangle,line width=2pt,align=left,below=.5cm of Alog] (Atra) {
$\hlf{\texttt{pub}}!\hms{\texttt{publish}}$
};
\node[draw,black,rectangle,line width=2pt,align=left,below=.5cm of Blog] (Btra) {
$\hlf{\texttt{bro}}?\hms{\texttt{subscribe}}$
};
\node[opacity=0,draw,black,rectangle,line width=2pt,align=left,below=.75cm of Clog] (anchor_transparent) {
$\hlf{\texttt{bro}}?\hms{\texttt{subscribe}}$\\
$\hlf{\texttt{bro}}?\hms{\texttt{publish}}$\\
$\hlf{\texttt{bro}}!\hms{\texttt{publish}}$
};
\draw[->,line width=2pt,black] (Alog) -- (Atra);
\draw[->,line width=2pt,black] (Blog) -- (Btra);
\end{tikzpicture}}
            \caption{Partial local observation\label{fig:architecture_partial}}
        \end{subfigure} 
    \end{minipage}
    \caption{A simple publish/subscribe example: architectures \& interaction model \label{fig:architecture}}
%\vspace*{-.5cm}
\end{figure}

\paragraph{Accepted multi-traces}
Following the terminology of \cite{constraint_based_oracles_for_timed_distributed_systems,a_small_step_approach_to_multi_trace_checking_against_interactions}, we call \emph{multi-trace} a collection of local traces, one per remote subsystem.
Fig.\ref{fig:architecture_local} depicts a multi-trace involving 3 local traces: 
$\hlf{\texttt{bro}}?\hms{\texttt{subscribe}}.\hlf{\texttt{bro}}?\hms{\texttt{publish}}.\hlf{\texttt{bro}}!\hms{\texttt{publish}}$ for subsystem $\hlf{\texttt{bro}}$,
$\hlf{\texttt{pub}}!\hms{\texttt{publish}}$ for $\hlf{\texttt{pub}}$, 
and 
$\hlf{\texttt{sub}}!\hms{\texttt{subscribe}}.\hlf{\texttt{sub}}?\hms{\texttt{publish}}$ for $\hlf{\texttt{sub}}$. 
It is possible to interleave these local traces to obtain the global trace in Fig.\ref{fig:architecture_global}, i.e. the multi-trace in Fig.\ref{fig:architecture_local} corresponds to the tuple of projections of the global trace in Fig.\ref{fig:architecture_global} onto each of the sub-systems. The tuple of projections of a global trace is unique. However, conversely, one might compute several global traces associated to the same tuple of local traces. This is because, in all generality, there is no ordering between actions occurring on different lifelines. For example, from the multi-trace of Fig.\ref{fig:architecture_local}, one could reconstruct the global trace:\\
\noindent 
\centerline{{\small $\hlf{\texttt{pub}}!\hms{\texttt{publish}}.\hlf{\texttt{sub}}!\hms{\texttt{subscribe}}.\hlf{\texttt{bro}}?\hms{\texttt{subscribe}}.\hlf{\texttt{bro}}?\hms{\texttt{publish}}.\hlf{\texttt{bro}}!\hms{\texttt{publish}}.\hlf{\texttt{sub}}?\hms{\texttt{publish}}$}}  
\noindent
The tuple of projections of this global trace is also the multi-trace in Fig.
\ref{fig:architecture_local}.
With the algorithm from \cite{a_small_step_approach_to_multi_trace_checking_against_interactions} one can recognize exactly accepted multi-traces (e.g. the one from Fig.\ref{fig:architecture_local}), which correspond to projections of accepted global traces (e.g. Fig.\ref{fig:architecture_global}).

\paragraph{Logging and Partial observability}
Offline RV requires to collect execution traces prior to their analyses.
In this process, it might be so that some subsystems cannot be equipped with observation devices. Moreover, due to the absence of synchronization between the local observations, the different logging processes might cease at uncorrelated moments.
For example, let us consider the multi-trace in Fig.\ref{fig:architecture_partial} as an observed execution of the system considered in Fig.\ref{fig:architecture}, where, by hypothesis, the subsystem $\hlf{\texttt{sub}}$ is not observed.
Remark that this multi-trace corresponds to a partial observation of the 
multi-trace in Fig.\ref{fig:architecture_local}. Indeed, each trace corresponding to a given subsystem in Fig.\ref{fig:architecture_partial} is a prefix of the trace corresponding to the same sub-system in Fig.\ref{fig:architecture_local}. 
Thus,
if $\hlf{\texttt{sub}}$ executions were also observed and with longer observation times for each local observation processes, it may well be that one would have observed the multi-trace in Fig.\ref{fig:architecture_local} rather than the one in Fig.\ref{fig:architecture_partial}. 
For that reason, when analysing the multi-trace in Fig.\ref{fig:architecture_partial} against the interaction in Fig.\ref{fig:example_interaction}, we need the RV process not to conclude on the occurrence of an error.
In fact, the multi-traces that we recognize as correct are those in which each of the local traces can be extended to reconstruct a multi-trace accepted by the interaction and we call them multi-prefixes of accepted multi-traces. 
%Thus, a major difficulty for the offline RV of DS is the consideration of partial observability, either because of unobserved subsystems or because of an early interruption of the observation of local traces. 
Let us remark that a projection of a prefix of an accepted global trace is a prefix of accepted multi-trace. However the reverse is not true. For example, there exists no prefix of a global trace accepted by the interaction in
Fig.\ref{fig:example_interaction} that projects on the multi-trace in 
Fig.\ref{fig:architecture_partial}.
This is because the emission of $\hms{\texttt{subscribe}}$ by $\hlf{\texttt{sub}}$ would precede its reception by $\hlf{\texttt{bro}}$ in any accepted global trace. However, this emission is not observed in the multi-trace in Fig.\ref{fig:architecture_partial}. 
Therefore, dealing with partial observability does not boil down to a simple adaptation of the algorithm in
\cite{a_small_step_approach_to_multi_trace_checking_against_interactions}.
In this paper, the aforementioned two types of partial observation (unobserved subsystems and early interruption of observation) will be approached in the same manner, noting in particular that an empty local trace can be seen both as missing and incomplete. The key mathematical operator used for that purpose consists in the removal of a lifeline from both interactions and multi-traces.
%In addition, some subsystems may not be observed via an instrumentation and it might also be that, due to a lack of synchronization between the local instrumentations, some
%Local traces of a multi-trace are local and partial views of the same global execution.
%some subsystems are not observed at all or that, due to issues of synchronisation of the instrumentation, events may be missing from a local trace.
This operator allows us to define an algorithm for recognizing multi-prefixes of accepted multi-traces while avoiding the complex search for a matching global execution, taking into account potential missing actions.

%A slight modification of the algorithm in \cite{a_small_step_approach_to_multi_trace_checking_against_interactions} would allow recognizing projections of prefixes of accepted global traces. However, this is not sufficient to tackle partial observability. Indeed, supposing the local observation device on $\hlf{\texttt{bro}}$ ceases to early and $\hlf{\texttt{bro}}!\hms{\texttt{publish}}$ is not observed in the multi-trace from Fig.\ref{fig:architecture_local}, then the resulting multi-trace could not be recognized because while $\hlf{\texttt{bro}}!\hms{\texttt{publish}}$ occurs at the end of the local component on $\hlf{\texttt{bro}}$ it does not occur at the end of a corresponding global execution.

\section{Multi-traces, interactions, and removal operations\label{sec:defs}}

\subsection{Multi-traces\label{sec:multitraces}}

As outlined in Section~\ref{sec:context}, a DS is a collection of communicating subsystems, each having a lifeline as local interface. 
Hence a DS is characterized by a finite set of lifelines $L \subseteq \mathcal{L}$, called a {\em signature}.
For $L \subseteq \mathcal{L}$, $\mathbb{A}(L)$ denotes the set $\cup_{l \in L} \mathbb{A}_l$.

The executions of a DS are then associated to {\em multi-traces} i.e. collections of traces, one per lifeline (see Definition \ref{def:multitrace}).

\begin{definition}\label{def:multitrace}
Given $L \subseteq {\cal L}$,
the set $\mathbb{M}(L)$ of {\em multi-traces over $L$} is\footnote{Given a family $(A_i)_{i \in I}$ of sets indexed by a finite set $I$, $\prod_{i \in I} A_i$ is the set of tuples $(a_1, \ldots, a_i,\ldots)$ with $\forall i \in I, a_i \in A_i$.} $\prod_{l \in L} \mathbb{T}_l$.\\ 
For $\mu = (t_l)_{l \in L}$ in $\mathbb{M}(L)$, we denote by $\mu_{|l}$ the trace component $t_l \in \mathbb{T}_l$ and by $\overline{\mu} = \{ \mu' ~|~\mu' \in \mathbb{M}(L), \forall l \in L, \mu'_{|l} \in \overline{\mu_{|l}}\}$ the set of its {\em multi-prefixes}.  
\end{definition}

Multi-prefixes are extended to sets: $\overline{M}$ is the set of all multi-prefixes of all multi-traces in $M \subseteq \mathbb{M}(L)$.
We denote by $\varepsilon_L$ the empty multi-trace in $\mathbb{M}(L)$ defined by $\forall l \in L, {\varepsilon_L}_{|l} = \varepsilon$.
Additionally, for any $\mu \in \mathbb{M}(L)$, we use the notations $\mu[t]_l$ to designate the multi-trace $\mu$ in which the component on $l$ has been replaced by $t \in \mathbb{T}_l$ and $|\mu|$ to designate the cumulative length $|\mu| = \sum_{l \in L} |\mu_{|l}|$ of $\mu$.

As discussed in Section \ref{sec:context}, two communication actions occurring on different traces of a multi-trace cannot be temporally ordered. 
Likewise, when several subsystems are observed concurrently, there is no way to synchronize the endings of their observations. So, any multi-trace $\mu' \in \overline{\mu}$ can be understood as a partial observation of the execution characterized by $\mu$.
An edge case of this partial observation occurs when some of the subsystems are not observed at all, i.e. when some lifelines are missing. The $\macroElim_h$ function of Definition~\ref{def:lifeline_removal} simply removes the trace concerning the lifeline $h$ from a multi-trace.

\begin{definition}\label{def:lifeline_removal}
For $L \subseteq {\cal L}$, the function $\macroElim_h: \mathbb{M}(L) \rightarrow \mathbb{M}(L\setminus \{h\})$ is s.t.:
$
\forall \mu \in \mathbb{M}(L), \;
\macroElim_h(\mu) = (\mu_{|l})_{l \in L\setminus \{h\}}
$
\end{definition}

The function $\macroElim_h$ is canonically extended to sets. We introduce operations to add an action to the left (resp. right) of a multi-trace. For the sake of simplicity, we use the same symbol ~\multiAppendLeft~  for these left- and right-concatenation operations: 
\[\forall a \in \mathbb{A}(L), \forall \mu \in \mathbb{M}(L), \;  a~\multiAppendLeft~\mu = \mu[a.\mu_{|\theta(a)}]_{\theta(a)} \mbox{~~~and~~~} \mu~\multiAppendRight~a= \mu[\mu_{|\theta(a)}.a]_{\theta(a)}\] 
Note that for any $\mu$ and $a$, we have $|\mu ~\multiAppendRight~ a| = |a ~\multiAppendLeft~ \mu| = |\mu|+1$.
We extend ~\multiAppendLeft~ to sets of multi-traces as follows:
$a ~\multiAppendLeft~ T = \{ a ~\multiAppendLeft~ \mu ~|~ \mu \in T \}$ and 
$T ~\multiAppendRight~ a = \{ \mu ~\multiAppendRight~ a ~|~ \mu \in T \}$.

Property \ref{lem:elimination_append} then trivially relates the ~\multiAppendLeft~ concatenation operation with the removal operation $\macroElim$.

\begin{property}[Removing lifelines and appending actions\label{lem:elimination_append}]
For $\mu \in \mathbb{M}(L)$ and $a \in \mathbb{A}(L)$, if $\theta(a) = h$ then $\macroElim_h(a ~\multiAppendLeft~ \mu) = \macroElim_h(\mu)$ and $\macroElim_h(\mu ~\multiAppendRight~ a) = \macroElim_h(\mu)$,
else $\macroElim_h(a ~\multiAppendLeft~ \mu) = a ~\multiAppendLeft~ \macroElim_h(\mu)$ and $\macroElim_h(\mu ~\multiAppendRight~ a) = \macroElim_h(\mu) ~\multiAppendRight~ a$.
\end{property}

For two multi-traces $\mu_1$ and $\mu_2$ in $\mathbb{M}(L)$:
\begin{itemize}
    \item $\mu_1 \cup \mu_2$ denotes the alternative defined as follows: $\mu_1 \cup \mu_2 = \{ \mu_1,\mu_2 \}$;
    \item $\mu_1 \multiSeq \mu_2$ denotes their sequencing defined as follows: if $\mu_2 = \varepsilon_L$ then $\mu_1 \multiSeq \mu_2 = \mu_1$  else, $\mu_2$ can be  written  as $a ~\multiAppendLeft~ \mu_2'$ and $\mu_1 \multiSeq \mu_2 = (\mu_1 ~\multiAppendRight~ a) \multiSeq \mu_2'$;  
   \item $\mu_1 \multiInterleaving \mu_2$ denotes their interleaving and is defined as the set of multi-traces describing parallel compositions of $\mu_1$ and $\mu_2$:
\[
\begin{array}{c}
\varepsilon_L \multiInterleaving \mu_2  = \{ \mu_2 \}  ~~~~~~~~~~~~~~~~~
\mu_1 \multiInterleaving \varepsilon_L  = \{ \mu_1 \} \\
(a_1 ~\multiAppendLeft~ \mu_1)  ~\multiInterleaving~  (a_2 ~\multiAppendLeft~ \mu_2) ~ = ~ ( a_1  ~ \multiAppendLeft ~  (\mu_1  ~\multiInterleaving~ ( a_2  ~\multiAppendLeft~ \mu_2 ))) \cup (a_2 ~\multiAppendLeft~ ((a_1 ~\multiAppendLeft~ \mu_1)  ~ \multiInterleaving ~  \mu_2 )))
\end{array}
\]
\end{itemize}
Let us remark that $\mu'$ is a prefix of a multi-trace $\mu$ (i.e. 
$\mu' \in \overline{\mu}$) iff there exists $\mu''$ verifying $\mu' ; \mu'' = \mu$.
Operations $\cup$, $\multiSeq$ and $\multiInterleaving$ are extended to sets of multi-traces as $\multiDiamond : \mathcal{P}(\mathbb{M}(L))^2 \rightarrow \mathcal{P}(\mathbb{M}(L))$ for $\multiDiamond \in \{ \cup,~\multiSeq,~\multiInterleaving \}$.
Operators $\multiSeq$ and $\multiInterleaving$ being associative, this allows for the definition of repetition operators in the same manner as the Kleene star is defined over the classical concatenation. Given $\multiDiamond \in \{ \multiSeq,~\multiInterleaving \}$, the Kleene closure $^{\multiDiamond *}$ is s.t. for any set of multi-traces $T \subseteq \mathbb{M}(L)$ we have:
\[ T^{\multiDiamond *}
= \bigcup_{\substack{j \in \mathbb{N}}} T^{\multiDiamond j} \mbox{ with } 
T^{\multiDiamond 0}
= \{ \varepsilon_L \} \mbox{ and }
T^{\multiDiamond j}
=
T \multiDiamond T^{\multiDiamond (j-1)}
\mbox{ for } j > 0
\]

$\mathbb{M}(L)$ fitted with the set of algebraic operators $\mathcal{F} = \{\cup, \multiSeq, \multiInterleaving, ~^{\multiSeq *}, ~^{\multiInterleaving *} \}$ is an $\mathcal{F}$-algebra.
The operation $\macroElim_h$ preserves the algebraic structures between the $\mathcal{F}$-algebras
of signatures $L$ and $L \setminus \{h\}$.

\begin{property}[Elimination preserves operators\label{lem:elimination_preserves_sched}]
For any $\mu_1$ and $\mu_2$ in $\mathbb{M}(L)$, for any $\multiDiamond \in \{\cup,\multiSeq,\multiInterleaving\}$, $\macroElim_h(\mu_1 \multiDiamond \mu_2) = \macroElim_h(\mu_1) \multiDiamond \macroElim_h(\mu_2)$.
\end{property}

\begin{proof}
%Obvious for $\cup$. By induction on $\mu_2$ for $\multiSeq$ and on  $\mu_1$ and $\mu_2$ for $\multiInterleaving$. See Appendix \ref{anx:proofs_multitraces}.

\noindent For $\multiDiamond = \cup$, $\macroElim_h(\mu_1 \cup \mu_2) = \macroElim_h(\{\mu_1,\mu_2\})$ by definition of the $\cup$ operator between multi-traces
and $\macroElim_h(\mu_1) \cup \macroElim_h(\mu_2) = \{ \macroElim_h(\mu_1) , \macroElim_h(\mu_2) \}$ by definition of the $\cup$ operator between multi-traces.

\noindent For $\multiDiamond = \multiSeq$, let us reason by induction on $\mu_2$:
\begin{itemize}
    \item $\macroElim_h(\mu_1 \multiSeq \varepsilon_L) = \macroElim_h(\mu_1) = \macroElim_h(\mu_1) \multiSeq  \varepsilon_{L'} = \macroElim_h(\mu_1) \multiSeq \macroElim_h(\varepsilon_L)$
    \item if $\mu_2 = a ~\multiAppendLeft~ \mu_2'$ then 
    \[
    \begin{array}{lclr}
    \macroElim_h(\mu_1 \multiSeq \mu_2) 
    &
    =
    &
    \macroElim_h((\mu_1 ~\multiAppendRight~a) \multiSeq \mu_2')
    &
    ~~~~\text{\scriptsize by definition of } \multiSeq
    \\
    &
    =
    &
    \macroElim_h((\mu_1 ~\multiAppendRight~a)) \multiSeq \macroElim_h(\mu_2')
    &
    ~~~~\text{\scriptsize by induction}
    \end{array}
    \]
    Then:
    \begin{itemize}
        \item if $\theta(a) = h$ we have:
        \[
        \begin{array}{lclr}
        \macroElim_h(\mu_1 \multiSeq \mu_2) 
        &
        =
        &
        \macroElim_h(\mu_1) \multiSeq \macroElim_h(\mu_2')
        &
        ~~~~\text{\scriptsize by Prop.\ref{lem:elimination_append}}
        \\
        &
        =
        &
        \macroElim_h(\mu_1) \multiSeq \macroElim_h(a ~\multiAppendLeft~\mu_2')
        &
        ~~~~\text{\scriptsize by Prop.\ref{lem:elimination_append}}
        \\
        &
        =
        &
        \macroElim_h(\mu_1) \multiSeq \macroElim_h(\mu_2)
        &
        \end{array}
        \]
        \item if $\theta(a) \neq h$ we have:
        \[
        \begin{array}{lclr}
        \macroElim_h(\mu_1 \multiSeq \mu_2) 
        &
        =
        &
        (\macroElim_h(\mu_1) ~\multiAppendRight~ a) \multiSeq \macroElim_h(\mu_2')
        &
        ~~~~\text{\scriptsize by Prop.\ref{lem:elimination_append}}
        \\
        &
        =
        &
        \macroElim_h(\mu_1) \multiSeq (a ~\multiAppendLeft~\macroElim_h(\mu_2'))
        &
        ~~~~\text{\scriptsize by definition of }\multiSeq
        \\
        &
        =
        &
        \macroElim_h(\mu_1) \multiSeq \macroElim_h(\mu_2)
        &
        ~~~~\text{\scriptsize by Prop.\ref{lem:elimination_append}}
        \end{array}
        \]
    \end{itemize}
\end{itemize}
\noindent For $\multiDiamond = \multiInterleaving$ we can reason similarly, using induction on both $\mu_1$ and $\mu_2$.
\qed 
\end{proof}

The results from Property \ref{lem:elimination_preserves_sched} can be extended to sets of multi-traces and imply that repetitions of those scheduling algebraic operators with their Kleene closures are also preserved by the elimination operator $\macroElim_h$.

\subsection{Interactions\label{sec:interactions}}

\begin{wrapfigure}{r}{0.55\textwidth}
\vspace*{-1.75cm}
    \centering
\input{figures/interaction_semantics}
    \caption{Semantics of example from Fig.\ref{fig:example_interaction}}
    \label{fig:sem_detail}
\vspace*{-.6cm}
\end{wrapfigure}

Interaction models, such as the one in Fig.\ref{fig:example_interaction} can be formalized as terms of an inductive language. \cite{revisiting_semantics_of_interactions_for_trace_validity_analysis,a_small_step_approach_to_multi_trace_checking_against_interactions} consider an expressive language with two sequencing operators, weak and strict, for ordering actions globally.
In the current paper, as only collections of remote local traces are considered, weak and strict sequencing can no longer be distinguished. This explains why we only consider a unique sequencing operator $seq$ in Definition \ref{def:interaction_language}.

\begin{definition}\label{def:interaction_language}
Given signature $L$, the set $\mathbb{I}(L)$ of interactions over $L$ is the set of ground terms built over the following symbols provided with arities in $\mathbb{N}$:
\begin{itemize}
    \item the empty interaction $\varnothing $ and any action $a$ in $\mathbb{A}(L)$  of arity 0;
    \item the two loop operators $loop_S$ and $loop_P$ of arity 1; 
    \item and the three operators $seq$, $par$ and $alt$ of arity 2.
\end{itemize}
\end{definition}

The interaction term of Fig.\ref{fig:example_interaction} is: \\
{\small
$seq(\, loop_S( seq(\hlf{\texttt{pub}}!\hms{\texttt{publish}},\hlf{\texttt{bro}}?\hms{\texttt{publish}})),$
$seq( seq(\hlf{\texttt{sub}}!\hms{\texttt{subscribe}},\hlf{\texttt{bro}}?\hms{\texttt{subscribe}}),$ \\
\noindent $loop_S( seq( seq(\hlf{\texttt{pub}}!\hms{\texttt{publish}},\hlf{\texttt{bro}}?\hms{\texttt{publish}}),$
$seq(\hlf{\texttt{bro}}!\hms{\texttt{publish}},\hlf{\texttt{sub}}?\hms{\texttt{publish}})
))))$.
}

The semantics of an interaction can be defined as a set of multi-traces in a denotational style by associating each syntactic operator with an algebraic counterpart. This is sketched out in Fig.\ref{fig:sem_detail} in which the semantics of the interaction in Fig.\ref{fig:example_interaction} is given. The denotational formulation, which is compositional, is defined in Definition \ref{def:algebraic_multi_trace_semantics} and illustrated in Fig.\ref{fig:example_interaction}.
%resulting from the sequencing of three subterms (the first loop, the passing of \texttt{\hms{subscribe}} and the second loop).

\begin{definition}[$\mathbb{M}$-semantics\label{def:algebraic_multi_trace_semantics}] 
Given $L \subseteq {\cal L}$, the multi-trace semantics $\sigma_{|L} : \mathbb{I}(L) \rightarrow \mathcal{P}(\mathbb{M}(L))$ 
is defined inductively using the following interpretations:
\begin{itemize}
    \item $\{ \varepsilon_L \}$ for $\varnothing$ and $\{a ~\multiAppendLeft~ \varepsilon_L\}$ for $a$ in $\mathbb{A}_L$;
    \item $\multiSeq^*$ (resp. $\multiInterleaving^*$) for loop operator $loop_S$ (resp. $loop_P$);
    \item $\multiSeq$ (resp. $\multiInterleaving$ and $\cup$) for binary operator $seq$ (resp. $par$ and $alt$). 
\end{itemize}
\end{definition}

Interactions can also be associated with an operational semantics in the style of Plotkin \cite{equivalence_of_denotational_and_operational_semantics_for_interaction_languages}. Its definition relies on two predicates denoted by $\downarrow$ and $\rightarrow$: for an interaction $i$, $i \downarrow$ states that $\varepsilon_L \in \sigma_{|L}(i)$
and $i \xrightarrow{a} i'$ states that all multi-traces of the form $a~\multiAppendLeft~\mu'$ with $\mu' \in \sigma_{|L}(i')$ are multi-traces of $\sigma_{|L}(i)$. This operational semantics is equivalent to the denotational formulation.

\begin{property}[Operational semantics\label{prop:operational_formulation_multitrace}]
There exist a predicate $\downarrow \subseteq \mathbb{I}(L)$ and a relation $\rightarrow \subseteq \mathbb{I}(L) \times \mathbb{A}(L) \times \mathbb{I}(L)$ such that, for any $i \in \mathbb{I}(L)$ and $\mu \in \mathbb{M}(L)$, the statement
$\mu \in \sigma_{|L}(i)$ holds iff it can be proven using the following two rules:
{
\centering
\begin{minipage}{3cm}
\begin{prooftree}
\AxiomC{$i \downarrow$}
\UnaryInfC{$\varepsilon_L \in \sigma_{|L}(i)$}
\end{prooftree}
\end{minipage}
\begin{minipage}{4cm}
\begin{prooftree}
\AxiomC{$\mu \in \sigma_{|L}(i')$}
\AxiomC{$i \xrightarrow{a} i'$}
\BinaryInfC{$a ~\multiAppendLeft~ \mu \in \sigma_{|L}(i)$}
\end{prooftree}
\end{minipage}\\
}
\end{property}

\begin{proof}
Proof available in \cite{equivalence_of_denotational_and_operational_semantics_for_interaction_languages} for a trace semantics and given in Appendix~\ref{anx:operational} for this multi-trace semantics.
\end{proof}

\begin{wrapfigure}{r}{0.5\textwidth}
\vspace*{-1.25cm}
    \centering
\begin{tikzpicture}
\node (int_full) at (0,0) {
\scalebox{.85}{
\begin{tikzpicture}
\node (int) at (0,0) {\includegraphics[height=3.75cm]{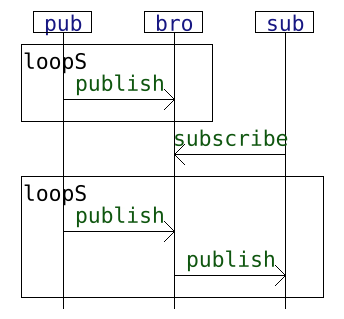}};
\begin{pgfonlayer}{main}
\fill[red!100,opacity=.6] (1,1.4) -- (1.65,1.4) -- (1.65,-1.7) -- (1,-1.7) -- cycle;
\end{pgfonlayer};
\end{tikzpicture}
}
};
\node[right=-.25cm of int_full] (int_hidden) {
\scalebox{.85}{
\includegraphics[height=3.75cm]{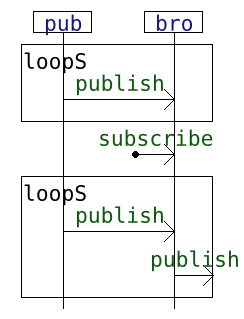}
}
};
\end{tikzpicture}
\vspace*{-.5cm}
    \caption{Removing lifeline $\hlf{\texttt{sub}}$\label{fig:hiding_example}}
\vspace*{-.5cm}
\end{wrapfigure}

The algebraic characterisation of Definition \ref{def:algebraic_multi_trace_semantics} underpins results involving the use of the $\macroElim_h$ function while the operational characterization of Property \ref{prop:operational_formulation_multitrace} is required in the definition and proof of the RV algorithm.
In this paper, we do not need the inductive definitions of $\downarrow$ and $\rightarrow$. It suffices to consider their existence (Property \ref{prop:operational_formulation_multitrace}). In addition, we will use the notation $i\xrightarrow{a}$ (resp. $i\not\xrightarrow{a}$) when there exists (resp. does not exist) an interaction $i'$ s.t. $i \xrightarrow{a} i'$. 

The removal of lifelines for multi-traces (cf. Definition \ref{def:lifeline_removal}) has a counterpart for interactions.
On the left of Fig.\ref{fig:hiding_example} we draw our previous example while highlighting lifeline $\hlf{\texttt{sub}}$ which we remove to obtain the interaction on the right. Whenever we remove a lifeline $l$, the resulting interaction does not contain any action occurring on $l$. Removal, as defined in\footnote{We overload the notation $\macroHide_h$ which applies to both multi-traces and interactions.} Definition~\ref{def:hiding} in functional style, preserves the term structure of interactions, replacing actions on the removed lifeline with the empty interaction.

\begin{definition}\label{def:hiding}
For a signature $L \subseteq {\cal L}$ and a lifeline $h \in L$ we define
%the function
$\macroHide_h: \mathbb{I}(L) \rightarrow \mathbb{I}(L \setminus \{h\})$ s.t. for any interaction $i \in \mathbb{I}(L)$:\\
$\macroHide_h(i) = \textbf{ match } i \textbf{ with}$\\
$\begin{array}{lll}
|~\varnothing & \rightarrow 
& \varnothing
\\
|~a \in \mathbb{A}(L) & 
\rightarrow 
& \textbf{if } \theta(a) = h \textbf{ then } \varnothing \textbf{ else } a
\\
|~f(i_1,i_2) & \rightarrow 
& f(\macroHide_h(i_1),\macroHide_h(i_2)) \textbf{ for } f \in \{ seq,alt,par \}
\\
|~loop_k(i_1) & \rightarrow 
& loop_k(\macroHide_h(i_1)) \textbf{ for } k \in \{S,P\}
\end{array}$
\end{definition}

\begin{theorem}[A property of multi-trace-semantics w.r.t. lifeline removal\label{th:semantics_of_hidden_interactions}]
For any signature $L$, any $i \in \mathbb{I}(L)$ and any $h \in L$: 
\[
\sigma_{|L \backslash \{ h \} }(\macroHide_h(i)) =
\macroElim_h(\sigma_{|L}(i))
\]
\end{theorem}

\begin{proof}
Let us reason by induction on the structure of interaction terms:
\begin{itemize}
    \item $\sigma_{|L \setminus \{h\}}(\macroHide_h(\varnothing)) = \sigma_{|L \setminus \{h\}}(\varnothing) = \{\varepsilon_{L \setminus \{h\}}\}
    = \macroElim_h(\{\varepsilon_{L}\}) = \macroElim_h(\sigma_{|L}(\varnothing))$
    \item for any $a \in \mathbb{A}(L)$ we have:
    \begin{itemize}
        \item if $\theta(a) = h$:
        \[
        \begin{array}{lcl}
        \sigma_{|L \setminus \{h\}}(\macroHide_h(a))
        &
        =
        &
        \sigma_{|L \setminus \{h\}}(\varnothing)
        \\
        &
        =
        &
        \{ \varepsilon_{L \setminus \{h\}} \}
        \\
        &
        =
        &
        \macroElim_h( \{ \varepsilon_{L} \} )
        \\
        &
        =
        &
        \macroElim_h( \{ a ~\multiAppendLeft~ \varepsilon_{L} \} )
        \\
        &
        =
        &
        \macroElim_h(\sigma_{|L}(a))
        \end{array}
        \]
        \item if $\theta(a) \neq h$:
        \[
        \begin{array}{lcl}
        \sigma_{|L \setminus \{h\}}(\macroHide_h(a))
        &
        =
        &
        \sigma_{|L \setminus \{h\}}(a)
        \\
        &
        =
        &
        \{ a ~\multiAppendLeft~ \varepsilon_{L \setminus \{h\}} \}
        \\
        &
        =
        &
        \macroElim_h( \{ a ~\multiAppendLeft~ \varepsilon_{L} \} )
        \\
        &
        =
        &
        \macroElim_h(\sigma_{|L}(a))
        \end{array}
        \]
    \end{itemize}
    \item with $(f,\multiDiamond) \in \{(seq,\multiSeq),~(par,\multiInterleaving),~(alt,\cup)\}$, for any $i_1,i_2$ in $\mathbb{I}(L)$:
    \[
    \begin{array}{lclr}
    \sigma_{|L \setminus \{h\}}(\macroHide_h(f(i_1,i_2)))
    &
    =
    &
    \sigma_{|L \setminus \{h\}}(f(\macroHide_h(i_1),\macroHide_h(i_2)))
    &
    \\
    &
    =
    &
    \sigma_{|L \setminus \{h\}}(\macroHide_h(i_1)) \multiDiamond \sigma_{|L \setminus \{h\}}(\macroHide_h(i_2))
    &
    \\
    &
    =
    &
    \macroElim_h(\sigma_{|L}(i_1)) \multiDiamond \macroElim_h(\sigma_{|L}(i_2))
    &
    \text{\scriptsize induction}
    \\
    &
    =
    &
    \macroElim_h( \sigma_{|L}(i_1) \multiDiamond \sigma_{|L}(i_2) )
    &
    \\
    &
    =
    &
    \macroElim_h( \sigma_{|L}(f(i_1,i_2)) )
    &
    \end{array}
    \]
    \item for any interaction $i$ and any $(k,\multiDiamond) \in \{(S,\multiSeq),~(P,\multiInterleaving)\}$:
    \[
    \begin{array}{lclr}
    \sigma_{|L \setminus \{h\}}(\macroHide_h(loop_k(i)))
    &
    =
    &
    \sigma_{|L \setminus \{h\}}(loop_k(\macroHide_h(i)))
    &
    \\
    &
    =
    &
    \sigma_{|L \setminus \{h\}}(\macroHide_h(i))^{\multiDiamond *}
    &
    \\
    &
    =
    &
    \macroElim_h(\sigma_{|L}(i))^{\multiDiamond *}
    &
    ~\text{\scriptsize induction}
    \\
    &
    =
    &
    \macroElim_h( \sigma_{|L}(i)^{\multiDiamond *} )
    &
    \\
    &
    =
    &
    \macroElim_h( \sigma_{|L}(loop_k(i)) )
    &
    \end{array}
    \]
\end{itemize}
\qed 
\end{proof}

%According to Theorem~\ref{th:semantics_of_hidden_interactions}, 
As, by construction, the order of removal of the lifelines does not matter, we generalize the notation $\macroHide_h$ with $\macroHide_{L'}$ to remove all lifelines of $L' \subseteq L$.

\section{Offline RV for multi-traces\label{sec:anahide_space}}

\begin{figure}[h]
\vspace{-.5cm}
    \centering
    \resizebox{.8\textwidth}{!}{\input{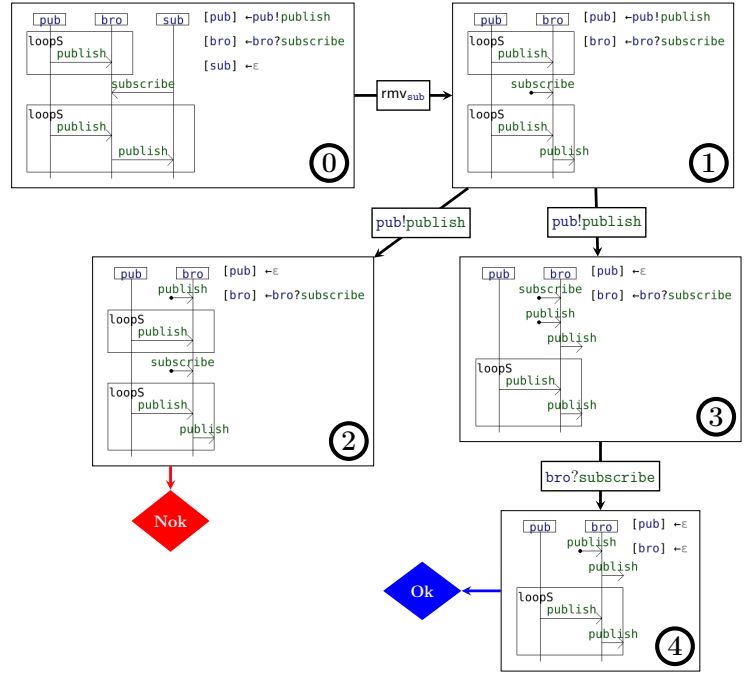}}
    \caption{An exploration s.t. $\omega_L(i,\mu) = Pass$}
    \label{fig:ana_example}
\vspace{-.5cm}
\end{figure}

Our goal is to define a process to analyze a multi-trace $\mu$, against a reference interaction $i$, both defined on a common signature $L$.
To check whether or not a multi-trace $\mu$ is accepted by $i$, i.e. $\mu \in \sigma_{|L}(i)$, the key principle given in~\cite{a_small_step_approach_to_multi_trace_checking_against_interactions} was to find a globally ordered behavior specified by $i$ (via the $\rightarrow$ execution relation) that matches $\mu$ i.e. an accepted global trace that can be projected into $\mu$. 
To do so, it relies on a general rule $(i,a ~\multiAppendLeft~ \mu') \leadsto (i',\mu')$ s.t. $i \xrightarrow{a} i'$ i.e. it explores all the actions $a$ directly executable from $i$ and that match the head of a local trace. The analysis is then pursued recursively from $(i',\mu')$ i.e. the multi-trace where $a$ has been removed and the follow-up interaction $i'$, until the multi-trace is emptied of actions.
For illustrative purposes, let us consider Fig.\ref{fig:ana_example} where each square annotated with a circled number (e.g. \textcircled{3}) contains an interaction drawn on the left and a multitrace on the right with one line for each of the 3 lifelines.  Starting from the interaction in \textcircled{3}, say $i_3$, with $ (\varepsilon,~\hlf{\texttt{bro}}?\hms{\texttt{subscribe}})$, one can see that we can reach \textcircled{4} by both consuming $\hlf{\texttt{bro}}?\hms{\texttt{subscribe}}$ from the multi-trace and executing it in $i_3$, leading to the interaction in \textcircled{4}, say $i_4$: based on $i_3 \xrightarrow{\hlf{\texttt{bro}}?\hms{\texttt{subscribe}}} i_4$, we have $(i_3, \hlf{\texttt{bro}}?\hms{\texttt{subscribe}} ~\multiAppendLeft~ (\varepsilon,\varepsilon)) \leadsto (i_4,(\varepsilon,\varepsilon))$. Thus, Fig.\ref{fig:ana_example} sketches the construction of a graph whose nodes are pairs of interactions and multitraces and whose arcs are built using the $\leadsto$ relation.

While in~\cite{a_small_step_approach_to_multi_trace_checking_against_interactions}, we were interested in solving the membership problem "$\mu \in \sigma_{|L}(i)$", we are now interested in defining an offline RV algorithm. In line with the discussion of Section~\ref{sec:context} about partial observability, $\mu$ reveals an error if $\mu$ is neither in $\sigma_{|L}(i)$ nor can be extended into an element of $\sigma_{|L}(i)$ i.e. $\mu$ diverges from $i$ iff $\mu \not\in \overline{\sigma_{|L}(i)}$.
In order to accommodate the need to identify prefixes of multi-traces, we 
%define a new algorithm following the philosophy of the one from \cite{a_small_step_approach_to_multi_trace_checking_against_interactions}, and which 
introduce a rule involving the removal operation. Indeed, as the execution relation $\rightarrow$ only allows executing actions in the global order in which they are intended to occur, we may reach cases in which the next action which may be consumed in the multi-trace cannot be executed due to having a preceding action missing in the multi-trace. Let us illustrate this with node \textcircled{0} of Fig.\ref{fig:ana_example}. 
$\hlf{\texttt{bro}}?\hms{\texttt{subscribe}}$ is the first action that occurs on lifeline $\hlf{\texttt{bro}}$ in the multi-trace. However, it cannot be executed because it must be preceded by $\hlf{\texttt{sub}}!\hms{\texttt{subscribe}}$. 
Yet, either because the behavior on lifeline $\hlf{\texttt{sub}}$ is not observed, or because the logging process ceased too early on $\hlf{\texttt{sub}}$, it might well be that $\hlf{\texttt{sub}}!\hms{\texttt{subscribe}}$ occurred in the actual execution although it was not logged.
With our new algorithm, because the condition that $\mu_{|\hlf{\texttt{sub}}} = \varepsilon$ is satisfied, from node \textcircled{0}, we apply a rule yielding the transformation $(i,\mu) \leadsto (\macroHide_{\hlf{\texttt{sub}}}(i),\macroHide_{\hlf{\texttt{sub}}}(\mu))$, removing lifeline $\hlf{\texttt{sub}}$, which allows us to pursue the analysis from node \textcircled{1}. To summarize, Fig.\ref{fig:ana_example} illustrates (part of) the graph that can be constructed from a pair $(i_0,\mu_0)$ using the relation $\leadsto$. We have $5$ nodes numbered from $0$ (the initial node of the analysis) to $4$. Arcs correspond either to the consumption of an action, or to the application of the $\macroHide$ operator, or to the emission of a verdict. The empty multi-trace in node \textcircled{4} allows us to conclude $\mu_0 \in \overline{\sigma_{|L}(i_0)}$. 

%The solution which we propose is based on transformations of the form $(i,\mu) \leadsto (i',\mu')$, where each node $(i,\mu)$ contains an interaction $i$ and a multi-trace to analyse $\mu$ and which are so that 
%$\mu' \in \overline{\sigma_{|L}(i')}$ implies $\mu \in \overline{\sigma_{|L}(i)}$. Given that for any $L \subset \mathcal{L}$ and $i \in \mathbb{I}_L$ we have $\varepsilon_L \in \overline{\sigma_{|L}(i)}$, it then suffices to find a path $(i,\mu) \overset{*}{\leadsto} (i',\varepsilon_{L'})$ (for a certain $L' \subseteq L$) to prove that $\mu \in \overline{\sigma_{|L}(i)}$. 

%where $i$ is an interaction and $\mu$ is a multi-trace for which it has to be decided whether or not $\mu\in \overline{\sigma_{|L}(i)}$.
%The connection of two vertices by an arc $(i, \mu)\leadsto (i', \mu')$ intuitively denotes that, in order to prove $\mu \in \overline{\sigma_{|L}(i)}$, it is sufficient to prove $\mu' \in \overline{\sigma_{|L}(i')}$. $i$ and $i'$ (respectively $\mu$ and $\mu'$) are not necessarily defined over the same signature as the definition of $i'$ (resp. $\mu'$) may result from the application of $\macroHide_{h}$ (resp. $\macroElim_{h}$) on $i$ (resp. $\mu$). Moreover, we also define two sink vertices $\macroOKVerdict$ and $\macroKOVerdict$.

\subsection{Search graph\label{ssec:graph}}

As the $\macroHide$ operator has the effect of changing the signature, we introduce the set $\mathbb{I}_{\mathcal{L}}$ (resp. $\mathbb{M}_{\mathcal{L}}$) to denote the set of all interactions (resp. multi-traces) defined on a signature of ${\mathcal{L}}$. 
Let us define a directed search graph with vertices either of the form $(i,\mu) \in \mathbb{I}_{\mathcal{L}} \times \mathbb{M}_{\mathcal{L}}$ or one of two specific verdicts $\macroOKVerdict$ and $\macroKOVerdict$.
We denote by $\mathbb{V}$ the set of all vertices:
\[\mathbb{V} = \{\macroOKVerdict,\macroKOVerdict\} \cup ( \; \bigcup_{ L \subseteq {\mathcal{L}}}  \mathbb{I}(L) \times \mathbb{M}(L) \; )\]
The arcs of $\mathbb{G}$ are defined by 4 rules: $\macroRulePass$, $\macroRuleFail$ leading to respectively the sink vertices $\macroOKVerdict$ and $\macroKOVerdict$, $\macroRuleExec$ (for "execute") for consuming an action of the multi-trace according to the $\rightarrow$ predicate of the operational formulation (cf. Property~\ref{prop:operational_formulation_multitrace}), and $\macroRuleHide$ (for "removal"), for removing  a lifeline from the interaction and multi-trace.

\begin{definition}[Search graph\label{def:multipref_anahide_rules}]
$\mathbb{G} = (\mathbb{V},\leadsto)$ is the graph s.t. for all $v,v'$ in $\mathbb{V}$, $v \leadsto v'$ iff there exists a rule $R_x$  with $x \in \{o,n,e,r\}$ s.t. $(R_x) \frac{v}{v'}$ where rules $R_x$ are defined as follows, with $L\subseteq \mathcal{L}$, $h\in L$, $i,i'\in \mathbb{I}(L)$, and $\mu,\mu'\in \mathbb{M}(L)$:

\noindent\begin{minipage}[l]{.4\textwidth}
\begin{prooftree}
\AxiomC{$i$\hspace{-.35cm}}
\AxiomC{$\varepsilon_{L}$}
\LeftLabel{($\macroRulePass$)}
\RightLabel{}
\BinaryInfC{$\macroOKVerdict$}
\end{prooftree}
\end{minipage}
\begin{minipage}[r]{.575\textwidth}
\begin{prooftree}
\AxiomC{$i$\hspace{-.35cm}}
\AxiomC{$\mu$}
\LeftLabel{($\macroRuleHide$)}
\RightLabel{$
\left.
\begin{array}{l}
\mu_{|h} = \varepsilon
\end{array}
\right.
$}
\BinaryInfC{$\macroHide_{h}(i) \hspace{0.5cm} \macroElim_{h}(\mu)$}
\end{prooftree}
\end{minipage}

\vspace*{.3cm}

\noindent\begin{minipage}[l]{.45\textwidth}
\begin{prooftree}
\AxiomC{$i$\hspace{-.35cm}}
\AxiomC{$\mu$}
\LeftLabel{($\macroRuleExec$)}
\RightLabel{$
\left\{\begin{array}{l}
\exists a \in \mathbb{A}(L),
\\
\mu = a  ~\multiAppendLeft~ \mu' \land i \xrightarrow{a} i'
\end{array}\right.$}
\BinaryInfC{$i' \hspace{0.5cm} \mu'$}
\end{prooftree}
\end{minipage}
\begin{minipage}[r]{.575\textwidth}
\begin{prooftree}
\AxiomC{$i$\hspace{-.35cm}}
\AxiomC{$\mu$}
\LeftLabel{($\macroRuleFail$)}
\RightLabel{$
\left\{
\begin{array}{l}
(\forall l \in L,\mu_{|l} \neq \varepsilon) \; \wedge \\
\left(\begin{array}{l}
\forall a\in \mathbb{A}(L),\forall \mu'\in \mathbb{M}(L),\\
\mu = a ~\multiAppendLeft~ \mu'\Rightarrow i\not\xrightarrow{a}
\end{array}\right)
\end{array}
\right.
$}
\BinaryInfC{$\macroKOVerdict$}
\end{prooftree}
\end{minipage}

\vspace*{.2cm}

\end{definition}

Rules $\macroRuleExec$ and $\macroRuleHide$ specify edges of the form $(i,\mu) \leadsto (i',\mu')$ with $i'$ and $\mu'$ defined on the same signature: the application of $\macroRuleExec$ corresponds to the simultaneous consumption of an action at the head of a component of $\mu$ and the execution of a matching action in $i$ while  the application of $\macroRuleHide$ corresponds to the removal of a lifeline $h$ s.t. $\mu_{|h} = \varepsilon$. Moreover vertices of the form $(i,\mu)$ are not sinks of $\mathbb{G}$. Indeed, if $\mu = \varepsilon_{L}$ then $\macroRulePass$ can apply, otherwise $\mu \neq \varepsilon_{L}$ and:
\textbf{(1)} if at least a component $\mu_{|h}$ of $\mu$ is empty, then rule $\macroRuleHide$ can apply. \textbf{(2)} if there is a match between an action that can be executed from $i$ and the head of a component of the multi-trace then rule $\macroRuleExec$ can apply. \textbf{(3)} if both conditions 1 and 2 do not hold then rule $\macroRuleFail$ applies.

%\begin{enumerate}
%    \item if at least a component $\mu_{|h}$ of $\mu$ is empty, then rule $\macroRuleHide$ can apply,
%    \item if there is a match between an action that can be executed from $i$ and the head of a component of the multi-trace then rule $\macroRuleExec$ can apply, 
%    \item if both conditions 1 and 2 do not hold then rule $\macroRuleFail$ applies.
%\end{enumerate}

Proving $\mu \in \overline{\sigma_{|L}(i)}$, amounts to exhibiting a path in $\mathbb{G}$ starting from $(i,\mu)$ and leading to the verdict $\macroOKVerdict$. Fig.\ref{fig:ana_example} depicts such a path for the multi-trace $\mu_0 =(\hlf{\texttt{pub}}!\hms{\texttt{publish}},\hlf{\texttt{bro}}?\hms{\texttt{subscribe}},\varepsilon)$
w.r.t. the interaction $i_0$ of node \textcircled{0}. A first step (application of $\macroRuleHide$) removes lifeline $\hlf{\texttt{sub}}$ leading to node \textcircled{1}. This is possible because $\mu_{|\hlf{\texttt{sub}}} = \varepsilon$. From there, by applying rule $\macroRuleExec$, the execution of 
$\hlf{\texttt{pub}}?\hms{\texttt{publish}}$ 
allows to reach either node \textcircled{2} or node \textcircled{3} depending on the loop used. From node \textcircled{3}, the previous removal of lifeline $\hlf{\texttt{sub}}$ has unlocked the execution of $\hlf{\texttt{bro}}?\hms{\texttt{subscribe}}$ (application of $\macroRuleExec$).
What remains is $\varepsilon_L$ and hence we can apply rule $\macroRulePass$. From the existence of this path leading to $\macroOKVerdict$ we conclude that $\mu_0$ is a prefix of a multi-trace of the interaction depicted in Fig.\ref{fig:example_interaction}.

\begin{property}[Finite search space\label{prop:finitesearchspace}]
Let $L \subseteq \mathcal{L}$, $\mu \in \mathbb{M}(L)$ and $i \in \mathbb{I}(L)$. The sub-graph of $\mathbb{G}$ of all vertices reachable from $(i,\mu)$ is finite. 
\end{property}

\begin{proof}
It follows from the following two observations
\textbf{(1)} any path in that sub-graph is finite
and \textbf{(2)} there is a finite number of paths.

The first point \textbf{(1)} can be proven using the following measure on vertices $s\in\mathbb{V}$ of $\mathbb{G}$: 
\[|s|=\left\{
\begin{array}{ll}
     0 & \text{if }s \in \{\macroOKVerdict,\macroKOVerdict\} \mbox{ (by convention)} \\
     |\mu|+|L| + 1 & \text{if }s = (i,\mu)\text{ with } L\subseteq \mathcal{L},\ i \in \mathbb{I}(L),\ \mu \in \mathbb{M}({L})
\end{array}
\right.
\]
For any transition $(i,\mu) \leadsto (i',\mu')$ in $\mathbb{G}$, we have $|(i',\mu')| = |(i,\mu)| - 1$ whether the rule that is applied is $\macroRuleExec$ or $\macroRuleHide$. Any other transition leads to either $\macroOKVerdict$ or $\macroKOVerdict$, which are sinks of $\mathbb{G}$. Hence, because $|(i,\mu)|$ is finite, positive, and decreases, any outgoing path from a node $(i,\mu)$ is finite before ultimately reaching either of $\macroOKVerdict$ or $\macroKOVerdict$.

The second point \textbf{(2)} comes from the fact that for any vertex $(i,\mu)$, there exists a finite number of outgoing transitions. Indeed, there can only be a finite number of possible applications of $\shortColOrange{R_e}$ because there cannot be more matches than the number of actions in $i$ and there cannot be more than $|L|$ different applications of $\shortColViolet{R_h}$ because there cannot be more than $|L|$ empty trace components on $\mu$.
%We can conclude that, for any vertex $(i,\mu)$ the sub-graph of $\mathbb{G}$ reachable from $(i,\mu)$ is finite.
\qed 
\end{proof}

%For any transition $(i,\mu) \leadsto (i',\mu')$ in $\mathbb{G}$, we have $|(i',\mu')| < |(i,\mu)|$ and from any vertex $(i,\mu)$, there exists a finite number of outgoing arcs.

Given our relation $\leadsto$ between vertices $\mathbb{V}$ of graph $\mathbb{G}$, for any two vertices $v,v' \in \mathbb{V}$, we denote\footnote{For any relation $\rightarrow \subset E^2$ on a set $E$, $\overset{*}{\rightarrow}$ is the reflexive and transitive closure of $\rightarrow$.} by $v \overset{*}{\leadsto} v'$ the existence of a path in $\mathbb{G}$ from $v$ to $v'$.

An interesting property of graph $\mathbb{G}$, related to the use of the $\macroRuleHide$ rule, is given in Property \ref{lem:confluence_anahide}.
It states that if, from a given vertex $(i,\mu)$, we can reach $\macroOKVerdict$ by any given means, then, if we can also apply rule $\macroRuleHide$ so that $(i, \mu) \leadsto (\macroHide_{h}(i),\macroElim_{h}(\mu))$ for any lifeline $h$, then we can also reach $\macroOKVerdict$ from $(\macroHide_{h}(i),\macroElim_{h}(\mu))$. 

This can be described as a property of confluence given that it states that we may take another path, in which we might as well hide lifeline $h$, so as to reach $\macroOKVerdict$.

\begin{property}[A property of the analysis graph\label{lem:confluence_anahide}]
For any $i \in \mathbb{I}(L)$, any $a \in \mathbb{A}(L)$, any $\mu \in \mathbb{M}(L)$ and any $h \in L$ we have:

\begin{prooftree}
\AxiomC{$(i, \mu) \overset{*}{\leadsto} \macroOKVerdict$}
\AxiomC{$(i, \mu) \leadsto (\macroHide_{h}(i),\macroElim_{h}(\mu))$}
\LeftLabel{}
\RightLabel{}
\BinaryInfC{$(\macroHide_{h}(i),\macroElim_{h}(\mu)) \overset{*}{\leadsto} \macroOKVerdict$}
\end{prooftree}
\end{property}

\begin{proof}
Let us reason by induction on the measure $|(i,\mu)|$:
\begin{itemize}
    \item If $|(i,\mu)|=1$ then $\mu = \varepsilon_{\emptyset}$ and the premise do not hold because we cannot apply $\macroRuleHide$
    \item If $|(i,\mu)| = 2$ then we must have $|\mu| = 0$ and $|L| = 1$ (the other case is not possible given that we can only have an empty multi-trace because $\mathbb{A}_\emptyset = \emptyset$). Then, we have $(i,\varepsilon_{L}) \leadsto (\macroHide_{h}(i),\varepsilon_{\emptyset})$ and we can immediately apply rule $\macroRulePass$ so that the conclusion holds
    \item If $|(i,\mu)| > 2$ then, if $|\mu| = 0$, we are in the same case as the previous one. Let us hence suppose that $|\mu| \geq 1$ which also implies that $|L| \geq 1$ so that $\mathbb{A}(L) \neq \emptyset$. Then, given $(i,\mu) \overset{*}{\leadsto} \macroOKVerdict$, we may have as a first transition in the path:
    \begin{itemize}
        \item either an application of $\macroRuleExec$ and in that case there exists $a$, $i'$ and $\mu'$ s.t. $\mu = a ~\multiAppendLeft~ \mu'$ and $i \xrightarrow{a} i'$ and we have
        $(i,\mu) \leadsto (i',\mu') \overset{*}{\leadsto} \macroOKVerdict$. Then:
        \begin{itemize}
            \item on the one hand we can apply the induction hypothesis on $(i',\mu')$ because we have that $(i', \mu') \leadsto (\macroHide_{h}(i'),\macroElim_{h}(\mu'))$ trivially still holds. Then we can conclude that $(\macroHide_{h}(i'),\macroElim_{h}(\mu')) \overset{*}{\leadsto} \macroOKVerdict$
            \item on the other hand, given $\mu = a ~\multiAppendLeft~ \mu'$, we must have $h \neq \theta(a)$ for the hypothesis $(i, \mu) \leadsto (\macroHide_{h}(i),\macroElim_{h}(\mu))$ to hold.
            Therefore if $a$ is executable in $i$ then it is also executable in $\macroHide_{h}(i)$ and we have $\macroHide_{h}(i) \xrightarrow{a} \macroHide_{h}(i')$ because $\macroHide_{h}$ is a homomorphism and hence preserves the algebraic structures of the IL. Also, we have that $\macroElim_{h}(\mu) = \macroElim_{h}(a ~\multiAppendLeft~ \mu') = a ~\multiAppendLeft~ \macroElim_{h}(\mu')$.
            This then implies that we can apply $\macroRuleExec$ from $(\macroHide_{h}(i),\macroElim_{h}(\mu))$ so that
            $(\macroHide_{h}(i),\macroElim_{h}(\mu)) \leadsto (\macroHide_{h}(i'),\macroElim_{h}(\mu'))$
        \end{itemize}
        The two points above allow to conclude that $(\macroHide_{h}(i),\macroElim_{h}(\mu)) \overset{*}{\leadsto} \macroOKVerdict$
        \item or an application of $\macroRuleHide$ and in that case there exists a lifeline $l \in L$ such that we have $(i,\mu) \leadsto (\macroHide_{l}(i),\macroElim_{l}(\mu)) \overset{*}{\leadsto} \macroOKVerdict$ and then:
        \begin{itemize}
            \item if $l=h$ we can immediately conclude
            \item if $l \neq h$ then we can remark that:
            \begin{itemize}
                \item firstly $(\macroHide_{l}(i),\macroElim_{l}(\mu)) \leadsto (\macroHide_{h}(\macroHide_{l}(i)),\macroElim_{h}(\macroElim_{l}(\mu)))$ and, given that we have decremented the measure by applying a first time $\macroRuleHide$, we can apply the induction hypothesis so that\\$(\macroHide_{h}(\macroHide_{l}(i)),\macroElim_{h}(\macroElim_{l}(\mu))) \overset{*}{\leadsto} \macroOKVerdict$
                \item secondly we can remark that $\macroHide_{h}(\macroHide_{l}(i)) = \macroHide_{l}(\macroHide_{h}(i))$\\
                and $\macroElim_{h}(\macroElim_{l}(\mu)) = \macroElim_{l}(\macroElim_{h}(\mu))$
                \item finally we have:
                \[
                \begin{array}{lcl}
                (\macroHide_{h}(i),\macroElim_{h}(\mu))
                &
                \leadsto
                &
                (\macroHide_{l}(\macroHide_{h}(i)),\macroElim_{l}(\macroElim_{h}(\mu)))
                \\
                &
                =
                &
                (\macroHide_{h}(\macroHide_{l}(i)),\macroElim_{h}(\macroElim_{l}(\mu)))
                \\
                &
                \overset{*}{\leadsto}
                &
                \macroOKVerdict
                \end{array}
                \]
                and hence the property holds
            \end{itemize}
        \end{itemize}
    \end{itemize}
\end{itemize}
\qed 
\end{proof}

\subsection{Verdict and conformity\label{ssec:algo}}

In Definition \ref{def:multipref_anahide_verdict}, we define the conformance of a multi-trace $\mu$ with regards to an interaction $i$ as the existence of a path $(i,\mu) \overset{*}{\leadsto} \macroOKVerdict$.

\begin{definition}[Multi-trace analysis\label{def:multipref_anahide_verdict}]
For any $L \subset \mathcal{L}$, we define $\omega_L : \mathbb{I}(L) \times \mathbb{M}({L}) \rightarrow \{Pass,Fail\}$ s.t. for any $i \in \mathbb{I}(L)$ and $\mu \in \mathbb{M}(L)$:
\begin{itemize}
    \item $\omega_L(i,\mu) = Pass$ iff $(i,\mu) \overset{*}{\leadsto} \macroOKVerdict$
    %there exists a path in $\mathbb{G}$ from $(i,\mu)$ to $\macroOKVerdict$
    \item $\omega_L(i,\mu) = Fail$ otherwise
\end{itemize}
\end{definition}

Given Property~\ref{prop:finitesearchspace}, Definition~\ref{def:multipref_anahide_verdict} is well founded insofar as the sub-graph of $\mathbb{G}$ issued from any pair $(i,\mu)$ of $\mathbb{V}$ is finite and all paths from $(i,\mu)$ can be extended until reaching a verdict ($\macroOKVerdict$ or $\macroKOVerdict$). 
In order to prove that the algorithm thus defined indeed identifies prefixes of accepted multi-traces, 
we need to prove that the existence of a path from $(i,\mu)$ to $\macroOKVerdict$ guarantees that $\mu$ is a prefix of a multi-trace of $i$, and that the non-existence of such a path guarantees that $\mu$ is not such a prefix.

This proof relies on an additional property given in Property \ref{lem:elimination_prefix_closure}, in which We relate the $\macroElim$ operator to prefix closure (in the sense of multi-traces).

\begin{property}[Elimination and prefix closure\label{lem:elimination_prefix_closure}]
For any multi-trace $\mu \in \mathbb{M}(L)$, any set of multi-traces $T \in \mathcal{P}( \mathbb{M}(L) )$ and any $h \in L$:
\[
\left(
\begin{array}{c}
( \macroElim_h(\mu) \in \macroElim_h(T) )~
\wedge~( \mu_{|h} = \varepsilon )
\end{array}
\right)
\Rightarrow 
( \mu \in \overline{T} )
\]
\end{property}

\begin{proof}
If $\macroElim_h(\mu) \in \macroElim_h(T)$ this means that $(\mu_{|l})_{l \in L'} \in \macroElim_h(T)$. Then, there must exist a multi-trace $\mu_0 \in T$ and a trace component 
%$t \in \mathbb{T}_{\{h\}}$ 
$t \in \mathbb{T}_h$ 
such that $\forall~l \in L'$, $\mu_{0|l} = \mu_{|l}$ and $\mu_{0|h} = t$. Let us then consider the multi-trace $\mu_1$ such that $\forall~l \in L'$, $\mu_{1|l} = \varepsilon$ and $\mu_{1|h} = t$.
We then have, because $\mu_{|h} = \varepsilon$, that $\mu_0 = \mu \multiSeq \mu_1$ and hence $\mu$ is a prefix (in the sense of multi-traces) of $\mu_0 \in T$. Therefore $\mu \in \overline{T}$.
\qed 
\end{proof}

The correctness of our algorithm, given in Theorem \ref{th:multipref_equates_pass} then follows from Property \ref{lem:confluence_anahide} and Property \ref{lem:elimination_prefix_closure}.

\begin{theorem}[Correctness\label{th:multipref_equates_pass}]
For any $i \in \mathbb{I}(L)$ and any $\mu \in \mathbb{M}(L)$:
\[
\left(
\begin{array}{c}
\mu \in \overline{\sigma_{|L}(i)}
\end{array}
\right)
\Leftrightarrow 
\left(
\begin{array}{c}
\omega_L(i,\mu) = Pass
\end{array}
\right)
\]
\end{theorem}

\begin{proof}
We use the following notation: $\overline{\sigma_{|L}}(i)$ denotes the set $\overline{\sigma_{|L}(i)}$ of prefixes of multi-traces of $i$.
Let us reason by induction on the measure $|(i,\mu)|$.
\begin{itemize}
    \item If $|(i,\mu)|=1$ then $L = \emptyset$, $\mu = \varepsilon_{\emptyset}$ and hence we have both $\omega_L(i,\mu) = Pass$ because rule $\macroRulePass$ immediately applies and $\mu \in \overline{\sigma_{|L}}(i)$ because the empty multi-trace $\varepsilon_{L}$ is in the prefix closure of any non-empty set of multi-traces.
    \item If $|(i,\mu)| \geq 2$ then:
    \begin{itemize}
        \item If there exists a lifeline $h$ s.t. $\mu_{|h} = \varepsilon$ then we can apply rule $\macroRuleHide$ and we have $(i,\mu) \leadsto (\macroHide_{h}(i),\macroElim_{h}(\mu))$ and then:
        \begin{itemize}
            \item[$\Rightarrow$] if $\mu \in \overline{\sigma_{|L}}(i)$ then, as per Th.\ref{th:semantics_of_hidden_interactions} we have $\macroElim_{h}(\mu) \in  \overline{\sigma_{|L'}}(\macroHide_{h}(i))$. Given that we have decremented the measure, we can apply the induction hypothesis which implies that $(\macroHide_{h}(i),\macroElim_{h}(\mu)) \overset{*}{\leadsto} \macroOKVerdict$. Then, by transitivity $(i,\mu) \overset{*}{\leadsto} \macroOKVerdict$ and hence $\omega_{L}(i,\mu) = Pass$
            \item[$\Leftarrow$] if $\omega_L(i,\mu) = Pass$ we have a path $(i,\mu) \overset{*}{\leadsto} \macroOKVerdict$ then we can apply Prop.\ref{lem:confluence_anahide} to obtain that $(\macroHide_{h}(i),\macroElim_{h}(\mu)) \overset{*}{\leadsto} \macroOKVerdict$ and we can then apply the induction hypothesis so that $\macroElim_{h}(\mu) \in  \overline{\sigma_{|L'}}(\macroHide_{h}(i))$. Then, as per Th.\ref{th:semantics_of_hidden_interactions} this equates 
            $\macroElim_{h}(\mu) \in \macroElim_{h}(\overline{\sigma_{|L}}(i))$.
            Then, given that $\mu_{|h} = \varepsilon$, we can apply Prop.\ref{lem:elimination_prefix_closure} to conclude that $\mu \in \overline{\overline{\sigma_{|L}}(i)} = \overline{\sigma_{|L}}(i)$
        \end{itemize}
        \item If there are no lifeline $h$ s.t. $\mu_{|h} = \varepsilon$ then:
        \begin{itemize}
            \item[$\Rightarrow$] if $\mu \in \overline{\sigma_{|L}}(i)$, then there exists $\mu_+$ s.t. $\mu \multiSeq \mu_+ \in \sigma_{|L}(i)$. Then, because $\mu \multiSeq \mu_+ \neq \varepsilon_L$, as per Prop.\ref{prop:operational_formulation_multitrace} there exists $a$, $i'$ and $\mu_*'$ s.t. $\mu \multiSeq \mu_+ = a ~\multiAppendLeft~ \mu_*'$ and $i \xrightarrow{a} i'$ and $\mu_*' \in \sigma_{|L}(i')$.
            Then, because, there is no empty trace component on $\mu$, action $a$ must be taken from $\mu$ and not from $\mu_+$. Therefore there exists $\mu'$ and $\mu_+'$ such that $\mu = a ~\multiAppendLeft~ \mu'$ and $\mu \multiSeq \mu_+ = a ~\multiAppendLeft~ \mu_*' = (a ~\multiAppendLeft~ \mu') \multiSeq \mu_+'$ and therefore $\mu' \multiSeq \mu_+' = \mu_*' \in \sigma_{|L}(i')$. Hence $\mu' \in \overline{\sigma_{|L}(i')} = \overline{\sigma_{|L}}(i')$. Then:
            \begin{itemize}
                \item on the one hand we can apply the induction hypothesis on $i'$ and $\mu'$ so that we have $(i',\mu') \overset{*}{\leadsto} \macroOKVerdict$
                \item on the other hand, the fact that $i \xrightarrow{a} i'$ and $\mu = a ~\multiAppendLeft~ \mu'$ allows us to apply rule $\macroRuleExec$ so that we have $(i,\mu) \leadsto (i',\mu')$
            \end{itemize}
            From the two last points we conclude by transitivity that $(i,\mu) \overset{*}{\leadsto} \macroOKVerdict$ and hence the property holds.
            \item[$\Leftarrow$] if $\omega_{L}(i,\mu) = Pass$ we have a path $(i,\mu) \overset{*}{\leadsto} \macroOKVerdict$ given that we cannot apply rule $\macroRuleHide$, the only possible first transition in this path is an application of rule $\macroRuleExec$ i.e. there must exists $a$, $i'$ and $\mu'$ s.t. $i \xrightarrow{a} i'$ and $\mu = a ~\multiAppendLeft~ \mu'$ and $(i,\mu) \leadsto (i',\mu') \overset{*}{\leadsto} \macroOKVerdict$. Then:
            \begin{itemize}
                \item on the one hand we can apply the induction hypothesis on $i'$ and $\mu'$ so that we have $\mu' \in \sigMultiPrefProjected{\check{L}}(i')$ which implies the existence of $\mu_+'$ such that $\mu' \multiSeq \mu_+' \in \sigma_{|L}(i')$
                \item on the other hand the fact that $i \xrightarrow{a} i'$ and $\mu' \multiSeq \mu_+' \in \sigma_{|L}(i')$, as per Prop.\ref{prop:operational_formulation_multitrace} this implies that $a ~\multiAppendLeft~ (\mu' \multiSeq \mu_+') \in \sigma_{|L}(i)$. 
                In particular, this implies that $\mu = a ~\multiAppendLeft~ \mu' \in \overline{\sigma_{|L}}(i)$
            \end{itemize}
        \end{itemize}
    \end{itemize}
\end{itemize}
\qed 
\end{proof}

\subsection{Complexity\label{ssec:complexity}}

The problem of recognizing correct multi-prefixes w.r.t. interactions is NP-hard (Property \ref{prop:multipref_nphard_complexity}).
In \cite{a_small_step_approach_to_multi_trace_checking_against_interactions}, the problem of determining whether or not $\mu \in \sigma_{|L}(i)$ has been proven to be NP-Hard via a reduction of the 1-in-3\,SAT problem (inspired by \cite{realizability_and_verification_of_msc_graphs}). In this paper we detail the reduction of a more general satisfiability problem : 3SAT. This problem is this time reduced into a problem of recognizing a multi-trace as a prefix of a behavior accepted by a certain interaction model i.e. the resolution of $\mu \in \overline{\sigma_{|L}(i)}$.

Let $X=\{x_1,\ldots,x_n\}$ be a finite set of Boolean variables. A literal $\ell$ is either a Boolean variable $x\in X$ or its negation $\neg{x}$. A 3\,CNF (Conjunctive Normal Form) formula is an expression of the form $\phi = C_1 \wedge \ldots \wedge C_j \wedge \ldots \wedge  C_k $ with every clause $C_j$ being a disjunction of three distinct literals. On the left of Fig.\ref{fig:3sat_red} is given, as an example, such a boolean expression $\phi$. The 3\,SAT problem is then to determine whether or not $\phi$ is satisfiable (whether or not there exists a variable assignment which sets all clauses in $\phi$ to $true$).

\begin{property}
\label{prop:multipref_nphard_complexity}
The problem of determining whether or not $\mu \in \overline{\sigma_{|L}(i)}$ is NP-hard.
\end{property}

\begin{proof}
Given a 3\,CNF formula $\phi$, with $|X| = n$ variables and $k$ clauses, we consider a set of lifeline $L=\{l_1,\ldots,l_k\}$ (a lifeline per clause), a unique message $m$, and the multi-trace $\mu_{\scriptscriptstyle \text{3\,SAT}}=(l_1?m,\ldots,l_k?m)$.

For any literal $\ell$, we build a multi-trace $\mu_{\ell}$ such that for any $l_j \in L$, if $\ell$ occurs in clause $C_j$ then $\mu_{\ell|l_j} = l_j?m$ and otherwise $\mu_{\ell|l_j} = \varepsilon$.
That is, every positive (resp. negative) occurrence of a variable $x$ in a clause $C_j$ is represented by an action $l_j?m$ in $\mu_x$ (resp. in $\mu_{\neg{x}}$). 
Let us then consider the set of multi-traces $T= (\{\mu_{x_1}\} \cup \{\mu_{ \neg{x_1}}\} );(\{\mu_{x_2}\} \cup \{\mu_{ \neg{x_2}}\} );\ldots;(\{\mu_{x_n}\} \cup \{\mu_{ \neg{x_n}}\} )$.
Every $\mu \in T$ corresponds to a variable assignment of the 3\,SAT problem. Indeed, to build $\mu$ either $\mu_{x}$ or $\mu_{\neg{x}}$ is selected (via $\cup$) in the definition of $T$, and not both. As $T$ is built using the sequencing (via $;$) of such alternatives for all variables, multi-traces in $T$ simulate all possible variable assignments (the search space for satisfying $\phi$).  
Because every clause contains three literals, one of which must be set to $true$, there is at least one literal $\ell \in \{x, \neg{x}\}$ in $C_j$ set to $true$. Hence $l_j?m \in \overline{\mu_{|l_j}}$. We remark that  $\mu_{|l_j}$ can be a sequence of such emissions $l_j?m$ if more than one literal is set to true in $C_j$. This reasoning can be applied to all the clauses i.e. $\forall j \in [1,k]$, $l_j?m \in \overline{\mu_{|l_j}}$ which implies that $\mu_{\scriptscriptstyle \text{3\,SAT}} \in \overline{\mu}$ and hence $\mu_{\scriptscriptstyle \text{3\,SAT}} \in \overline{T}$.
Given that $T$ is equivalent to the semantics of an interaction $i$ of the form $seq(alt(i_{x_1},i_{\neg{x_1}}),seq(alt(i_{x_2},i_{\neg{x_2}}),\cdots,alt(i_{x_n},i_{\neg{x_n}})\cdots))$, with, for any literal $\ell$, $i_\ell$ being the sequencing of all $l_j?m$ such that $\ell$ appears in $C_j$, solving the 3\,SAT problem equates to solving $\mu_{\scriptscriptstyle \text{3\,SAT}} \in \overline{\sigma_{|L}(i)}$.
\qed 
\end{proof}

Hence, we have provided a polynomial reduction of 3\,SAT to the problem of recognizing multi-prefixes of accepted multi-traces. The reduction of the problem on the left of Fig.\ref{fig:3sat_red} is represented on the right of Fig.\ref{fig:3sat_red}, via drawing the resulting interaction and multi-trace.
The problem has 3 variables and 4 clauses. In the corresponding interaction, lifeline $l_1$ corresponds to the first clause and we see that it has a $l_1?m$ in the right branch of the first alternative, corresponding to $\neg x_1$, the right branch of the second for $\neg x_2$ and the right branch of the third for $\neg x_3$. The same applies to $l_2$, $l_3$ and $l_4$.

\begin{figure}[h]
    \centering
\begin{tikzpicture}
\node[draw,rectangle,line width=2pt] (problem_3sat) at (0,0) {
    $\begin{aligned}
    ~~~& ( \neg x_1 ~\vee \neg x_2 ~\vee \neg x_3 )\\
    \wedge~ & ( \neg x_1 ~\vee \phantom{\neg}x_2 ~\vee \phantom{\neg}x_3 )\\
    \wedge~ & ( \phantom{\neg}x_1 ~\vee \neg x_1 ~\vee \phantom{\neg}x_2 )\\
    \wedge~ & ( \phantom{\neg}x_2 ~\vee \phantom{\neg}x_3 ~\vee \neg x_3 )
    \end{aligned}$
};
\node[draw,rectangle,line width=2pt,below right=-2.5cm and .75cm of problem_3sat] (problem_multipref) {
    \begin{tikzpicture}
        \node[draw,rectangle] (int) at (0,0) {
            \includegraphics[width=4.5cm]{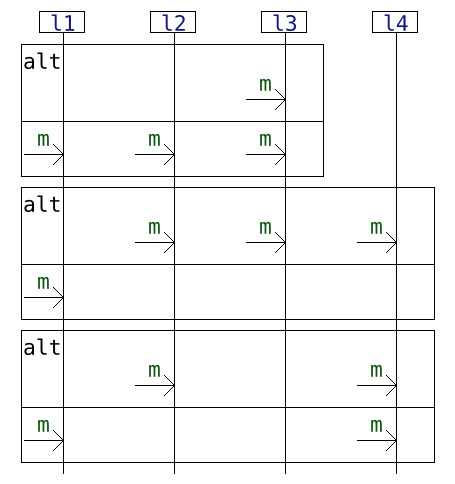}
        };
        \node[draw,rectangle,right=.5cm of int] (mu) {
            $
            \left\{~
            \begin{aligned}
            \relax[\hlf{l_1}] & \leftarrow & \hlf{l_1}?\hms{m}\\
            \relax[\hlf{l_2}] & \leftarrow & \hlf{l_2}?\hms{m}\\
            \relax[\hlf{l_3}] & \leftarrow & \hlf{l_3}?\hms{m}\\
            \relax[\hlf{l_4}] & \leftarrow & \hlf{l_4}?\hms{m}
            \end{aligned}
            \right.
            $
        };
    \end{tikzpicture}
};
\node[draw,rectangle,red,below=1cm of problem_3sat] (red) {
    \large Reduction
};
\draw[dashed,red,->] (problem_3sat) -- (red);
\draw[dashed,red,->] (red) -- (problem_multipref.west |- red);
\end{tikzpicture}
    \caption{Principle of 3\;SAT reduction}
    \label{fig:3sat_red}
\end{figure}
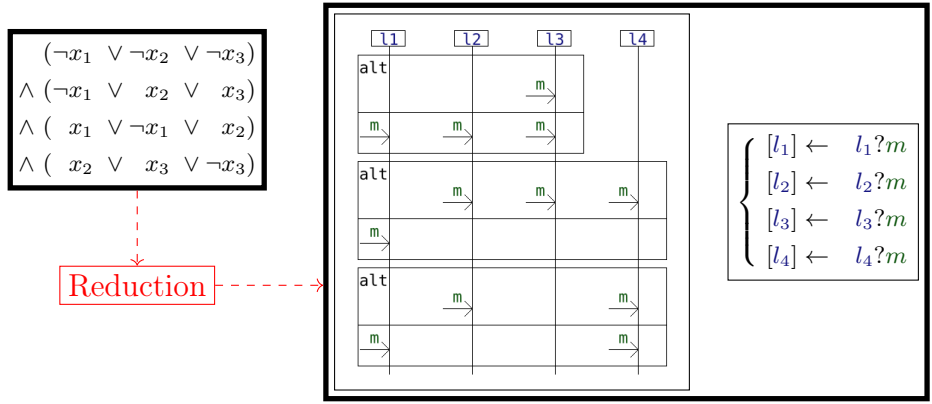

Given the NP-hardness of the underlying problem, the implementation of our algorithm, which is defined as a graph exploration of the search space $\mathbb{G}$, is combined with heuristic techniques to reduce the average complexity. Such techniques may include means to cut parts of the graph, the use of pertinent search strategies, of priorities and criteria for the selection of the next node to explore, or to further condition the use of the algorithm's rules.
%We have implemented our approach as an extension of the tool HIBOU \cite{hibou_label}. HIBOU takes the form of a command-line tool with a text-based input language.
%To prove $\mu \in \overline{\sigma_{|L}(i)}$ it suffices to find a path starting at $(i,\mu)$ and leading to $\macroOKVerdict$. Hence our algorithm relies on a partial exploration of $\mathbb{G}$, which can be stopped once $\macroOKVerdict$ is reached. Such a graph traversal can be done using a queue $Q$. As a result its performances can be improved via the use of search heuristics. A Depth First Search is a better choice than a Breadth First Search for most non trivial models (i.e. selecting child nodes before siblings). 
%Priorities and criteria for the selection of the next node from $Q$ may be considered. 
For instance, if $\macroRuleHide$ is applicable from a node $(i,\mu)$, we can apply $\macroElim$ on all lifelines which can be removed at the same time. Also, if $\macroRuleExec$ is applicable from that same node, we can choose not to apply it. Those two points are justified by properties of commutativity for $\macroElim$ and of a confluence/Church-Rosser property for relation $\leadsto$ (see Property \ref{lem:confluence_anahide}). 

We have implemented our approach as an extension of the tool HIBOU \cite{hibou_label} (a command-line tool with a text-based input language).
Various such techniques, not detailed here for lack of space,  are implemented in the tool.

\section{Experimental assessment\label{sec:experiments}}

In the following, we seek to evaluate our implementation (in HIBOU version 0.8.0).
In Section~\ref{ssec:3SAT benchmarks}, we use it to solve 3SAT problems and in Section~\ref{ssec:usecase}, we apply it on some practical examples from the literature.

\subsection{3\;SAT benchmarks}
\label{ssec:3SAT benchmarks}

In light of Property~\ref{prop:multipref_nphard_complexity}, we have experimented with the use of HIBOU for solving 3\,SAT problems via an automatic translation towards multi-trace analysis. The reduction and experiments are resp. detailed in Section \ref{ssec:complexity} and Appendix \ref{anx:exp_3sat}. \cite{hibou_3sat_bench} hosts the code to reproduce the experiments. 

We have compared the results HIBOU obtained on translated 3\,SAT problems against those of a SAT solver (Varisat \cite{varisat_rs}).
As input data we have used 3 sets of problems: two custom benchmarks with randomly generated problems
and the UF20 benchmark \cite{sat_benchmarks}. 

Fig.\ref{fig:hibou_varisat_3sat} provides details on 2 benchmarks with, on the top left, information about the input problems (numbers of variables, clauses, instances), on the bottom left statistical information about the time required for the analysis using each tool, and, on the right a corresponding scatter plot. In the plot, each point corresponds to a given 3-SAT problem, with its position corresponding to the time required to solve it (by Varisat on the $x$ axis and HIBOU on the $y$ axis). Points in red are unsatisfiable problems while those in blue are satisfiable. 

Those experiments underwrite the correctness of our implementation and provide an ad-hoc assessment of its performances. Let us keep in mind that our approach is not designed to solve 3\;SAT by contrast to dedicated 3\;SAT solvers.
%With both benchmarks we can observe that HIBOU systematically returns the correct result whether the problem is satisfiable or unsatisfiable. Overall HIBOU times are acceptable compared to the off-the-shelf solver for either verdict ($Fail$ for the custom benchmark or $Pass$ for both benchmarks), especially that the interactions are issued from combinatorial 3SAT instances. 

\begin{figure}[h]
    \centering
    \begin{subfigure}[t]{.975\linewidth}
        \centering
        \begin{minipage}{.49\linewidth}
\centering
{
\scriptsize  

\begin{tabular}{|l|r|}
\hline 
\# variables & 3-10 \\
\hline 
\# clauses & 4-50 \\
\hline 
\# instances & 663 \\
\hline 
\# SAT & 376 \\
\hline 
\# UNSAT & 287 \\
\hline 
\end{tabular}

~\\~\\

\begin{tabular}{|c|l|l|}
\cline{2-3}
\multicolumn{1}{c|}{} & \multicolumn{1}{c|}{varisat} & \multicolumn{1}{c|}{hibou}\\
\hline 
min & 0.01699 & 0.0002379 \\
\hline 
q1 & 0.01792 & 0.0012984 \\
\hline 
Mdn & 0.01806 & 0.0027920 \\
\hline 
M & 0.01833 & 0.0043448 \\
\hline 
q3 & 0.01848 & 0.0053158 \\
\hline 
max & 0.02892 & 0.0267174 \\
\hline 
$\sigma$ & 0.001017846 & 0.004637261 \\
\hline 
\end{tabular}
}       
\end{minipage} 
        \begin{minipage}{.49\linewidth}
            \centering
            \includegraphics[width=\textwidth]{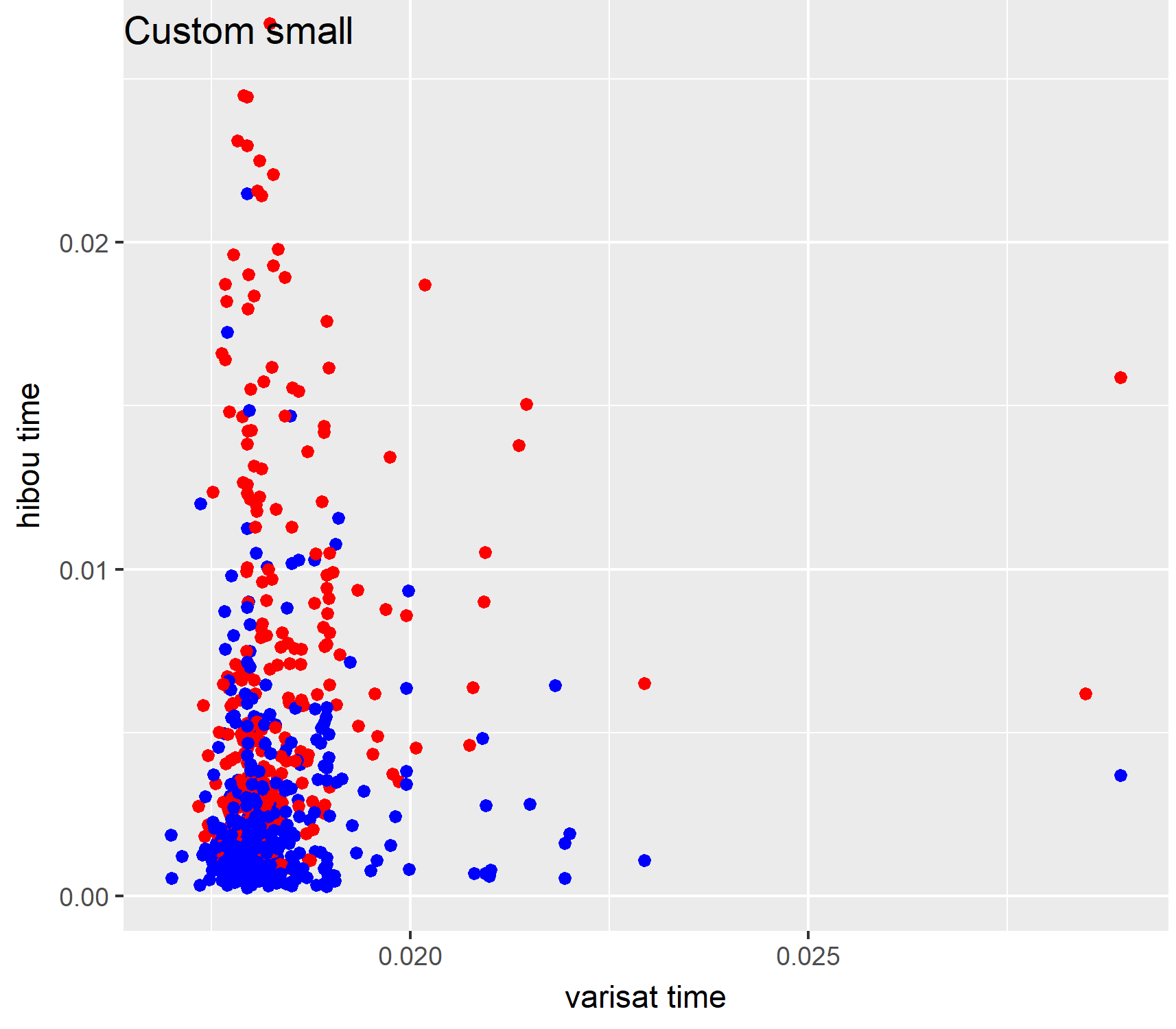}
        \end{minipage}
        \caption{Input problems and output results for 'small' custom benchmark\label{fig:Custom_small}}
    \end{subfigure}
    \begin{subfigure}[t]{.975\linewidth}
        \centering
        \begin{minipage}{.49\linewidth}
\centering
{
\scriptsize  

\begin{tabular}{|l|r|}
\hline 
\# variables & 20 \\
\hline 
\# clauses & 91 \\
\hline 
\# instances & 1000 \\
\hline 
\# SAT & 1000 \\
\hline 
\# UNSAT & 0 \\
\hline 
\end{tabular}

~\\~\\

\begin{tabular}{|c|l|l|}
\cline{2-3}
\multicolumn{1}{c|}{} & \multicolumn{1}{c|}{varisat} & \multicolumn{1}{c|}{hibou}\\
\hline 
min & 0.01559 & 0.007638 \\
\hline 
q1 & 0.01667 & 0.091421 \\
\hline 
Mdn & 0.01833 & 0.229745 \\
\hline 
M & 0.01847 & 0.313901 \\
\hline 
q3 & 0.01929 & 0.462385 \\
\hline 
max & 0.03989 & 1.666777 \\
\hline 
$\sigma$ & 0.00255181 & 0.2865485 \\
\hline 
\end{tabular}
}       
\end{minipage}
        \begin{minipage}{.49\linewidth}
            \centering
            \includegraphics[width=\textwidth]{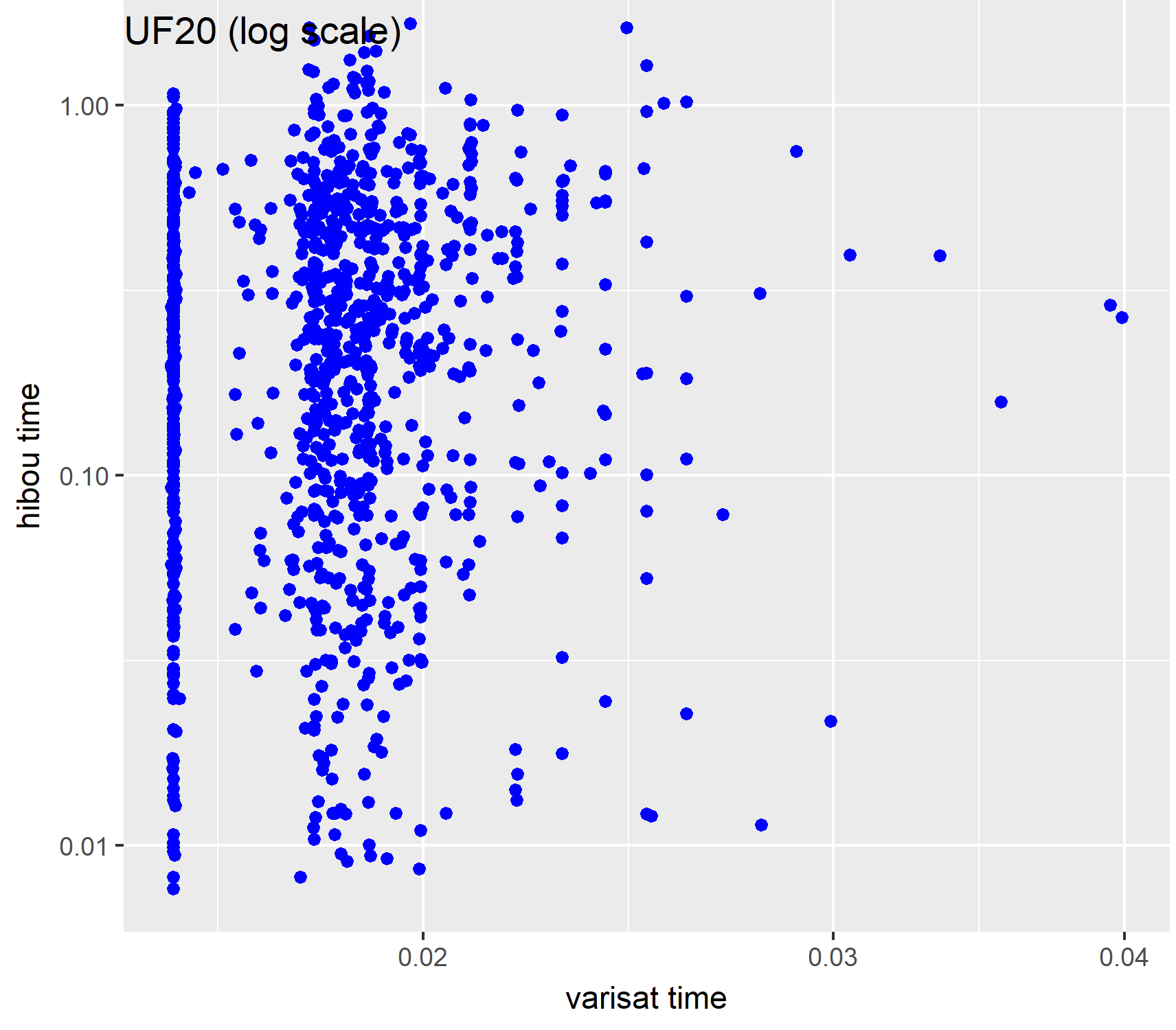}
        \end{minipage}
        \caption{Input problems and output results for UF-20 benchmark\label{fig:UF20}}
    \end{subfigure}
    \caption{Experiments on 3SAT benchmarks (times in seconds)\label{fig:hibou_varisat_3sat}}
    \vspace{-.5cm}
\end{figure}

\subsection{Use cases experiments}
\label{ssec:usecase}
So as to consider more concrete and varied interactions, we experiment with the following 4 examples:
a protocol for purchasing books \cite{coping_with_bad_agent_interaction_protocols_when_monitoring_partially_observable_multiagent_systems_AnconaFFM18},
a system for querying complex sensor data \cite{a_dynamic_and_context_aware_semantic_mediation_service_for_discovering_and_fusion_of_heterogeneous_sensor_data},
the Alternating Bit Protocol \cite{high_level_message_sequence_charts}
and a network for uploading data to a server \cite{comprehensive_multiparty_session_types}.
%We remark that for the e-commerce application~\cite{coping_with_bad_agent_interaction_protocols_when_monitoring_partially_observable_multiagent_systems_AnconaFFM18}, some intermediate parties may be not accessible for observation, therefore it is natural to disregard their components in the input multi-traces. 
Fig.\ref{fig:exp_usecases} partially reports on those experiments. More details are available in Appendix \ref{anx:exp_usecases} and online \cite{hibou_hiding_usecases}. For each example, we generated random accepted multi-traces (ACPT) up to some depth, for which we then randomly selected prefixes (PREF). For each such prefix we then performed mutations of three kinds: swapping actions (SACT), swapping trace components (SCMP) and inserting noise (NOIS). We report for each category of multi-traces times to compute verdicts in Fig.\ref{fig:exp_usecases}. As expected, running the algorithm on those multi-traces allows recognizing prefixes and mutants which go out of specification.

%For the experiments, we just apply random generation of multi-trace prefixes.    
%For validating our algorithm and its implementation, 
%For each interaction, we have generated a finite set of accepted multi-traces using a trace generation feature of the HIBOU tool, and then, for every such multi-trace, we computed a number of multi-prefixes.
%The experiments are summarised on Fig.\ref{fig:exp_usecases}. 
%The last four columns correspond to the four use cases. 
%On the third line are the number of accepted multi-traces generated for each interaction. Criteria to stop the exploration of the interaction are depicted in the second line. The fourth line gives the total number of multi-prefixes generated for each use case (between 20 and 30 multi-prefixes per generated multi-trace) and the last line gives the minimum and maximum sizes (in number of actions) of the generated multi-prefixes. For all generated multi-prefixes, HIBOU returns a $Pass$ verdict, in an almost immediate computation times, because the longest multi-trace considered is around 200.

%The four input use cases are detailed in Appendix \ref{anx:exp_usecases}. More details and all the related material is available online \cite{hibou_hiding_usecases}.

\begin{figure}[h]
    \centering
    \begin{minipage}{.49\linewidth}
       \begin{subfigure}[t]{.975\linewidth}
            \centering
            \includegraphics[width=\textwidth]{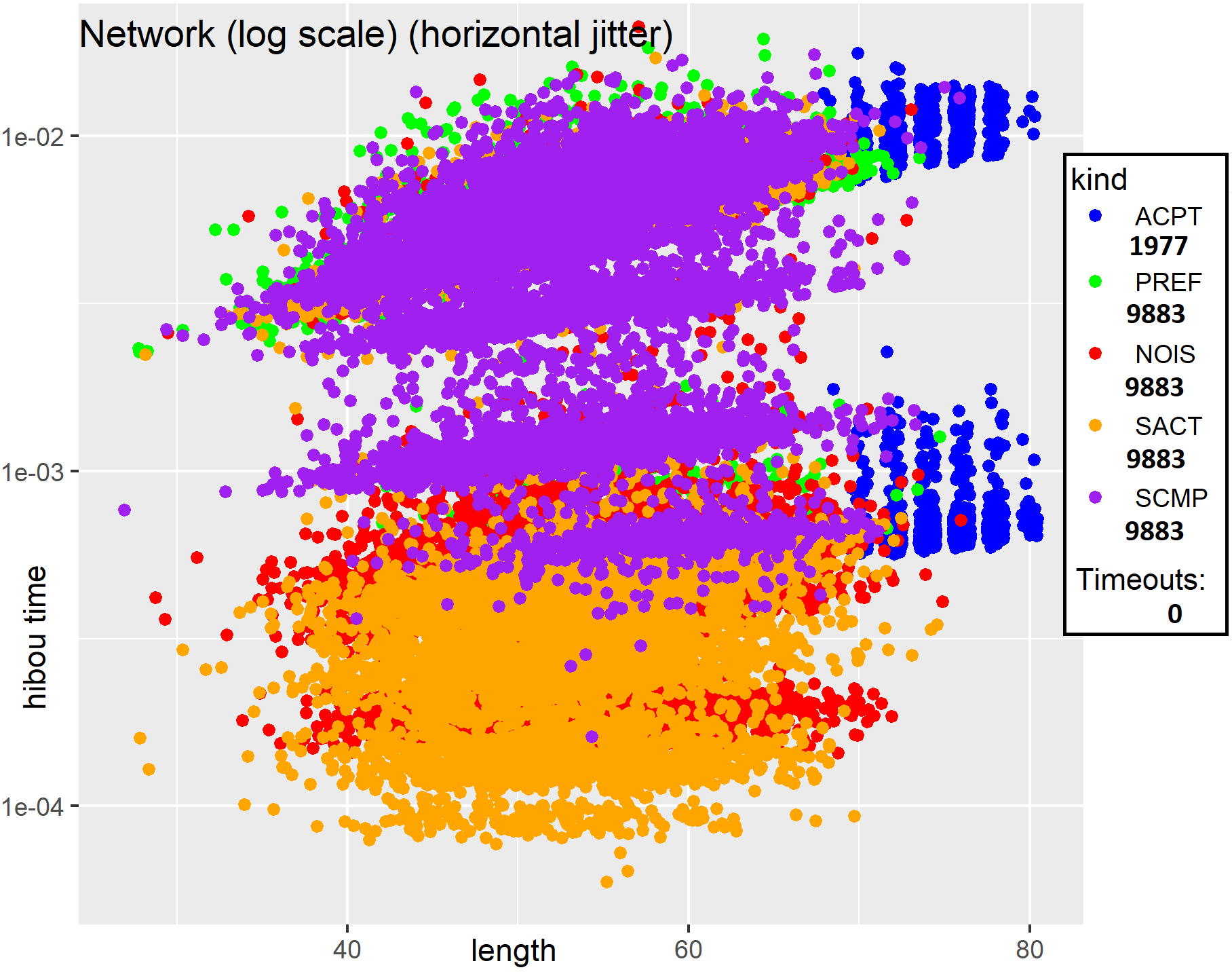}
            \caption{Network \cite{comprehensive_multiparty_session_types}}
        \end{subfigure} 
    \end{minipage}    
    \begin{minipage}{.49\linewidth}
    \begin{subfigure}[t]{.975\linewidth}
            \centering
            \includegraphics[width=\textwidth]{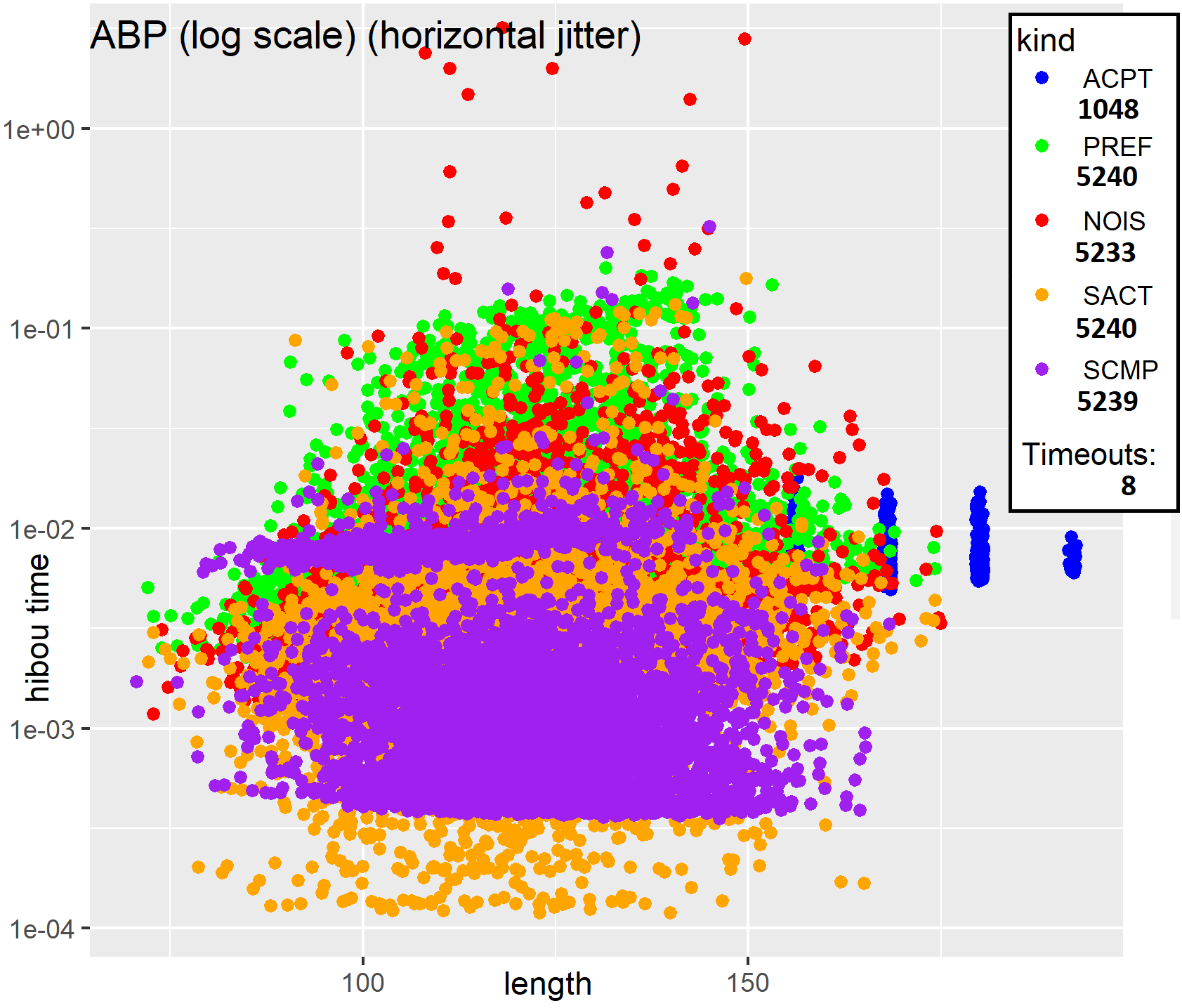}
            \caption{ABP \cite{high_level_message_sequence_charts}}
        \end{subfigure}
    \end{minipage}
    \caption{Experimental data on a selection of use cases (times in seconds)}
    \label{fig:exp_usecases}
    \vspace{-.8cm}
\end{figure}

\section{Related works\label{sec:related}}

Solutions to the oracle problem (offline RV) for DS using local logs often rely on a preliminary reordering of events using either timestamps \cite{passive_conformance_testing_of_service_choreographies} or some happened-before relations (of Lamport \cite{Lamport19b}) \cite{pivot_tracing_dynamic_causal_monitoring_for_distributed_systems,falcon_a_practical_log_based_analysis_tool_for_distributed_systems,constraint_based_oracles_for_timed_distributed_systems}.
In \cite{dioco_HieronsMN08,conformance_relations_for_distributed_testing_based_on_CSP_CavalcantiGH11,scenarios_based_testing_of_systems_with_distributed_ports} such solutions rely on a set of discrete and local behavioral models. DS behaviors are modeled by Input/Output Transition Systems (IOTS)~\cite{dioco_HieronsMN08,scenarios_based_testing_of_systems_with_distributed_ports} or by Communicating Sequential Processes (CSP)~\cite{conformance_relations_for_distributed_testing_based_on_CSP_CavalcantiGH11} and local observations are intertwined to associate them with global traces that can be analyzed w.r.t. models.
Those approaches however require to synchronize local observations, based on the states in which each of the logging processes terminates (e.g., based on quiescence states in~\cite{dioco_HieronsMN08}, termination/deadlocks in~\cite{conformance_relations_for_distributed_testing_based_on_CSP_CavalcantiGH11} or pre-specified synchronization points in~\cite{scenarios_based_testing_of_systems_with_distributed_ports}). 
%Without synchronization at the end of the observations (e.g., a  reception can be logged, unlike its corresponding emission), the analysis can't conclude the validity of the DS execution to the global traces of the reference specification. %d'où l'intérêt de considérer les prefixes...
%This limits the space of the multi-traces that can be processed.
The works~\cite{Hierons14,passive_conformance_testing_of_service_choreographies,monitoring_networks_through_multiparty_session_types_BocchiCDHY17,InckiA18} focus on verifying distributed executions against models of interaction (while \cite{Hierons14,InckiA18} concern MSC, \cite{passive_conformance_testing_of_service_choreographies} considers choreographic languages, \cite{monitoring_networks_through_multiparty_session_types_BocchiCDHY17} session types and \cite{coping_with_bad_agent_interaction_protocols_when_monitoring_partially_observable_multiagent_systems_AnconaFFM18} trace expressions). \cite{Hierons14,passive_conformance_testing_of_service_choreographies} propose offline RV that relies on synchronization hypotheses and on reconstructing a global trace by ordering events occurring at the distributed interfaces (by exploiting the observational power of testers~\cite{Hierons14} or timestamp information assuming clock synchronisation~\cite{passive_conformance_testing_of_service_choreographies}). Our RV approach for multi-traces does not require synchronization prerequisites on DS logging. Thus, unlike previous works on offline RV, we can analyze DS executions without the need for a synchronisation hypothesis on the ending of local observations. For online RV, the work~\cite{InckiA18} depends on a global component (network sniffer) while the work~\cite{monitoring_networks_through_multiparty_session_types_BocchiCDHY17} proposes local RV against projections of interactions satisfying conditions that enforce intended global behaviors.%( possibly combined with %messages flow %analysis through a global safe router for stronger ¨%conformance)
By contrast to these works we process collections of local logs against interactions. The work~\cite{coping_with_bad_agent_interaction_protocols_when_monitoring_partially_observable_multiagent_systems_AnconaFFM18} focuses on how distributed monitors can be adapted for partial observation. Yet, our notion of partial observation is distinct from that of \cite{coping_with_bad_agent_interaction_protocols_when_monitoring_partially_observable_multiagent_systems_AnconaFFM18} where messages are exchanged via channels which are associated to an observability likelihood. \cite{coping_with_bad_agent_interaction_protocols_when_monitoring_partially_observable_multiagent_systems_AnconaFFM18} uses trace expressions as specifications and proposes transformations that can adapt those expressions to partial observation by removing or making optional a number of identified unobservable events. We instead deal with partial observability from the perspective of analyzing truncated multi-traces due to synchronization issues.

To address design issues, we can also mention that early works~\cite{realizability_and_verification_of_msc_graphs,pattern_matching_and_membership_for_hierarchical_message_sequence_charts} considered checking basic MSCs against HMSC (High-level MSC, which are graphs of MSCs) as an MSC membership problem. Roughly speaking, a basic MSC equates a multi-trace and may specify a desired or unwanted scenario.
%, so answering membership can help either fix an error in the design or avoid redundancy.
%In those works, 
Some MSCs are marked as accepting within an HMSC, and a basic MSC belongs to the semantics of the graph iff it fully covers accepting (finite) paths in the graph. Thus, partially observed multi-traces cannot be assessed against HMSC, which does not answer the RV problem under observability limitations. %These works stated the NP-hardness of MSC membership by reduction from the 1-in-3-SAT problem~\cite{1in3SAT78}. 
%This computational cost suggests %an optimization aspect
%a heuristic approach that is considered in the context of the RV problem addressed in this paper.  
%-------------------
Logical properties have been widely used in (online) RV as reference specifications, in particular using the Linear Temporal Logic (LTL) whose semantics are generally given in the form of sets of traces. \cite{SenVAR04} extends a variant of LTL for which formulas relate to a subsystem and what it knows about the other subsystem' local states. It considers a collection of decentralized observers that share information about the subsystem executions that affect the validity of the formula. In other works \cite{Falcone16,El-HokayemF17}, the properties are expressed at the (global) system level and are transformed to decentralized observers, using LTL formula rewriting, so that there is no need for a global verifier gathering all information on the system's execution. 
%In~\cite{El-HokayemF17}, global properties are transformed to enable decentralized verification.
Unlike logics, interactions which are particularly adapted for specifying DS are barely used in RV (see the specification part of the taxonomy of RV tools~\cite{a_taxonomy_for_classifying_runtime_verification_tools_FalconeKRT21}).

\section{Conclusion}
We have proposed offline RV for multi-traces, i.e., sets of local execution logs collected on the DS. These multi-traces are partial views of the DS execution either because some components are not observed or because observations ceased early on some others. We check multi-traces against interactions (akin to UML-SD/MSC).
%, reflected by the desynchronized logging or unavailable logging at some subsystems.
We have proved the correctness of our offline RV algorithm that boils down to a graph search algorithm either by matching actions of the interaction against those of the input multi-trace or by applying the removal operations on multi-traces and interactions. Removal steps allow dealing with observability issues by enabling us to disregard
%by applying removal steps disregarding 
no longer observed parts of the interaction.
  % on interaction concerning some subsystems 
%Our implementation and experiments % on paper running example and a use case issued from literature
%show noticeable benefits of the local analyses, particularly in the case of non-conformance. 
Future works include other uses of the removal operator and investigating online RV.

\bibliographystyle{splncs04}
\bibliography{main}

\clearpage 

\appendix

\section{Operational formulation of the semantics (Section~\ref{sec:interactions}) \label{anx:operational}}

With Prop.\ref{prop:operational_formulation_multitrace}, we state the existence of an operational formulation of the algebraic multi-trace semantics from Def.\ref{def:algebraic_multi_trace_semantics}. 
In this appendix we will provide one such formulation complete with a definition and a proof of equivalence. The formulation relies on the definition of two inductive predicates: a termination predicate $\downarrow$ and an execution relation $\rightarrow$. 

The demonstration below mimics that given in \cite{equivalence_of_denotational_and_operational_semantics_for_interaction_languages} involving a denotational semantics defined as sets of global traces. In our case, there is one less sechulding operator and the denotational semantics is defined with sets of multi-traces.

\subsection{Termination}

If an interaction can express the empty multi-trace $\varepsilon_L$ then it means that it can immediately terminate i.e. that it is able to not express anything anymore. The problem of whether or not an interaction $i$ can immediately terminate can be answered systematically via the analysis of the term structure of $i$.
We provide a solution in the form of the termination predicate "$\downarrow$" given on Def.\ref{def:termination_predicate}. The formulation of that predicate is inspired from process algebras as in \cite{high_level_message_sequence_charts,operational_semantics_for_msc}.

The $\downarrow$ predicate can be inferred inductively from the term structure of interactions:
\begin{itemize}
    \item naturally the empty interaction $\varnothing$ only accepts $\varepsilon_L$, and can only terminate. As a result, we have $\varnothing \downarrow$
    \item any loop accepts $\varepsilon_L$ because it is possible to repeat zero times its content. Therefore, for any $i \in \mathbb{I}(L)$, and any $k \in \{S,P\}$ we have $loop_k(i)\downarrow$
    \item for interactions of the form $alt(i_1,i_2)$, if either $i_1$ or $i_2$ terminates then $alt(i_1,i_2)$ terminates
    \item for interactions of the form $f(i_1,i_2)$ with $f$ being a scheduling constructor ($seq,par$) it is required that both $i_1$ and $i_2$ terminate for $f(i_1,i_2)$ to terminate
\end{itemize}

\begin{definition}[Termination "$\downarrow$" predicate\label{def:termination_predicate}]
We define inductively the predicate $\downarrow \subset \mathbb{I}(L)$ such that for any two interactions $i_1$ and $i_2$ from $\mathbb{I}(L)$, for any $f \in \{seq,par\}$ and for any $k \in \{S,P\}$ we have:

{
\centering
\begin{minipage}{1.5cm}
\begin{prooftree}
\AxiomC{\phantom{$\top$}}
\UnaryInfC{$\varnothing \downarrow$}
\end{prooftree}
\end{minipage}
\begin{minipage}{2.5cm}
\begin{prooftree}
\AxiomC{$i_1 \downarrow$}
\UnaryInfC{$alt(i_1,i_2) \downarrow$}
\end{prooftree}
\end{minipage}
\begin{minipage}{2.5cm}
\begin{prooftree}
\AxiomC{$i_2 \downarrow$}
\UnaryInfC{$alt(i_1,i_2) \downarrow$}
\end{prooftree}
\end{minipage}
\begin{minipage}{2.75cm}
\begin{prooftree}
\AxiomC{$i_1 \downarrow$}
\AxiomC{$i_2 \downarrow$}
\BinaryInfC{$f(i_1,i_2) \downarrow$}
\end{prooftree}
\end{minipage}
\begin{minipage}{2.5cm}
\begin{prooftree}
\AxiomC{\phantom{$\top$}}
\UnaryInfC{$loop_k(i_1) \downarrow$}
\end{prooftree}
\end{minipage}
\vspace*{.1cm}
}
\end{definition}

The termination predicate $\downarrow$ characterizes the fact that an interaction can express the empty multi-trace $\varepsilon_L$ and therefore that it is in its semantics. As a result we formulate and prove this in Lem.\ref{lem:sem_de_terminates}.

\begin{lemma}[Characterization of termination w.r.t. $\sigma_{|L}$\label{lem:sem_de_terminates}]
For any $i \in \mathbb{I}(L)$:
\[
(i \downarrow ) \Leftrightarrow (\varepsilon \in \sigma_{|L}(i))
\]
\end{lemma}

\begin{proof}
Let us prove the equivalence of both predicate by induction on the term structure of $i$.
\begin{itemize}
    \item If $i = \varnothing$ the empty interaction, then we have both $\varnothing\downarrow$ and $\varepsilon_L \in \sigma_{|L}(\varnothing)$.
    \item If $ i \in \mathbb{A}(L)$, we have neither $i \downarrow$ nor $\varepsilon_L \in \sigma_{|L}(i)$.
    \item Let us now suppose that $i$ is of the form $seq(i_1,i_2)$, with $i_1$ and $i_2$ two sub-interactions that satisfy the induction hypotheses $(i_1 \downarrow ) \Leftrightarrow (\varepsilon_L \in \sigma_{|L}(i_1))$ and $(i_2 \downarrow ) \Leftrightarrow (\varepsilon_L \in \sigma_{|L}(i_2))$.
    \begin{itemize}
        \item[$\Leftarrow$] Let us suppose that $\varepsilon_L \in \sigma_{|L}(i)$. By definition of $\sigma_{|L}$ for the $seq$ constructor, this implies the existence of $\mu_1 \in \sigma_{|L}(i_1)$ and $\mu_2 \in \sigma_{|L}(i_2)$ such that $\varepsilon_L \in (\mu_1 \multiSeq \mu_2)$. This implies that $\mu_1 = \varepsilon_L$ and $\mu_2 = \varepsilon_L$. We can therefore apply the induction hypotheses, to obtain that $i_1 \downarrow$ and $i_2 \downarrow$. This in turn means that $seq(i_1,i_2)\downarrow$ by definition of the termination predicate.
        \item[$\Rightarrow$] Reciprocally, if $seq(i_1,i_2) \downarrow$, this means that both $i_1\downarrow$ and $i_2\downarrow$. As per the induction hypotheses, this means that $\varepsilon_L \in \sigma_{|L}(i_1)$ and $\varepsilon_L \in \sigma_{|L}(i_2)$. Therefore $\varepsilon_L \in \sigma_{|L}(i)$.
    \end{itemize}
    \item For interactions of the form $par(i_1,i_2)$, the reasoning is the same as for the previous case except that we use properties on the operator $\globalInterleaving$.
    \item Let us now suppose that $i$ is of the form $alt(i_1,i_2)$, with $i_1$ and $i_2$ two sub-interactions that satisfy the induction hypotheses.
    \begin{itemize}
        \item[$\Leftarrow$] Let us suppose that $\varepsilon_L \in \sigma_{|L}(i)$. By definition of $\sigma_{|L}$, this means that either $\varepsilon_L \in \sigma_{|L}(i_1)$ or $\varepsilon_L \in \sigma_{|L}(i_2)$ or both.
        Let us suppose that it is in $\sigma_{|L}(i_1)$ (the other cases can be treated similarly).
        As per the induction hypothesis, we therefore have $i_1\downarrow$. Then, by definition of the termination predicate, this implies that given that $alt(i_1,i_2) \downarrow$.
        \item[$\Rightarrow$] Reciprocally, if $alt(i_1,i_2) \downarrow$, this means that either $i_1\downarrow$ or $i_2\downarrow$ (or both). Let us suppose we have $i_1\downarrow$. As per the induction hypothesis, this means that $\varepsilon_L \in \sigma_{|L}(i_1)$. Therefore $\varepsilon_L \in \sigma_{|L}(i)$.
    \end{itemize}
    \item Let us finally consider the case where $i$ is of the form $loop_k(i_1)$, with $k \in \{S,P\}$. By definition, we always have $i\downarrow$ and $\varepsilon_L \in \sigma_{|L}(i)$.
\end{itemize}
\qed 
\end{proof}

\subsection{Execution relation \& operational-style semantics}

We define an execution relation $\rightarrow$ for our interaction language in Def.\ref{def:interaction_execution_relation}.

\begin{definition}[Execution relation\label{def:interaction_execution_relation}]
We define the execution relation $\rightarrow \subset \mathbb{I}(L) \times \mathbb{A}(L)  \times \mathbb{I}(L)$ such that for any action $a \in \mathbb{A}(L)$, for any interactions $i$, $i_1$, $i_1'$, $i_2'$ in $\mathbb{I}(L)$:

%act
\begin{minipage}{6cm}
\begin{prooftree}
\AxiomC{\phantom{$\xrightarrow{a}$}}
\UnaryInfC{$a \xrightarrow{a} \varnothing$}
\end{prooftree}
\end{minipage}

\vspace{0.1cm}

% alt  
\begin{minipage}{6cm}
\begin{prooftree}
\AxiomC{$i_1 \xrightarrow{a} i'_1$}
\UnaryInfC{$alt(i_1,i_2) \xrightarrow{a} i'_1$}
\end{prooftree}
\end{minipage}
\begin{minipage}{6cm}
\begin{prooftree}
\AxiomC{$i_2 \xrightarrow{a} i'_2$}
\UnaryInfC{$alt(i_1,i_2) \xrightarrow{a} i'_2$}
\end{prooftree}
\end{minipage}

\vspace{0.1cm}

% par
\begin{minipage}{6cm}
\begin{prooftree}
\AxiomC{$i_1 \xrightarrow{a} i'_1$}
\UnaryInfC{$par(i_1,i_2) \xrightarrow{a} par(i'_1,i_2)$}
\end{prooftree}
\end{minipage}
\begin{minipage}{6cm}
\begin{prooftree}
\AxiomC{$i_2 \xrightarrow{a} i'_2$}
\UnaryInfC{$par(i_1,i_2) \xrightarrow{a} par(i_1,i'_2)$}
\end{prooftree}
\end{minipage}

\vspace{0.1cm}

% seq  
\begin{minipage}{6cm}
\begin{prooftree}
\AxiomC{$i_1 \xrightarrow{a} i'_1$}
\UnaryInfC{$seq(i_1,i_2) \xrightarrow{a} seq(i'_1,i_2)$}
\end{prooftree}
\end{minipage}
\begin{minipage}{6cm}
\begin{prooftree}
\AxiomC{$i_2 \xrightarrow{a} i'_2$}
\RightLabel{$i_1 \downarrow$}
\UnaryInfC{$seq(i_1,i_2) \xrightarrow{a} i'_2$}
\end{prooftree}
\end{minipage}

\vspace{0.1cm}

\begin{minipage}{6cm}
\begin{prooftree}
\AxiomC{$i_1 \xrightarrow{a} i_1'$}
\UnaryInfC{$loop_S(i_1) \xrightarrow{a} seq(i_1',loop_S(i_1))$}
\end{prooftree}
\end{minipage}
\begin{minipage}{6cm}
\begin{prooftree}
\AxiomC{$i_1 \xrightarrow{a} i_1'$}
\UnaryInfC{$loop_P(i_1) \xrightarrow{a} par(i_1',loop_P(i_1))$}
\end{prooftree}
\end{minipage}

\vspace{0.1cm}

\end{definition}

This execution relation defines, for any interaction, which of its actions can be executed, and, if so, which interactions may result from those executions. This constitutes the "small-step" of a small-step operational semantics which we define in Def.\ref{def:interaction_operational_semantics}.

\begin{definition}[Operational semantics\label{def:interaction_operational_semantics}]
For any signature $L$, we define $\sigma_{o|L} : \mathbb{I}(L) \rightarrow \mathcal{P}(\mathbb{M}(L))$ by:

\begin{minipage}{4cm}
\begin{prooftree}
\AxiomC{$i \downarrow$}
\UnaryInfC{$\varepsilon_L \in \sigma_{o|L}(i)$}
\end{prooftree}
\end{minipage}
\begin{minipage}{6cm}
\begin{prooftree}
\AxiomC{$\mu \in \sigma_{o|L}(i')$}
\AxiomC{$i \xrightarrow{a} i'$}
\BinaryInfC{$a ~\multiAppendLeft~ \mu \in \sigma_{o|L}(i)$}
\end{prooftree}
\end{minipage}

\end{definition}

In the following, we will prove that this operational formulation is equivalent to the denotational formulation from Def.\ref{def:algebraic_multi_trace_semantics} i.e. that for any $i \in \mathbb{I}(L)$ we have $\sigma_{o|L}(i) = \sigma_{|L}(i)$, which justifies Prop.\ref{prop:operational_formulation_multitrace}. Note that, unlike Prop.\ref{prop:operational_formulation_multitrace}, we take care here to give another name to the operational semantics (by adding a subscript $o$ to $\sigma$). The introduction of a second notation makes it much easier to prove the equivalence of the two semantics by double inclusion.

\subsection{Left inclusion}

\begin{lemma}[Characterization (left side) of $\rightarrow$ w.r.t. $\sigma_{|L}$\label{lem:sem_de_execute1}]
For any action $a \in \mathbb{A}(L)$, for any multi-trace $\mu \in \mathbb{M}(L)$ and for any interactions $i$ and $i'$ from $\mathbb{I}(L)$:
\[
\left(
\begin{array}{c}
(i \xrightarrow{a} i')
\wedge 
(\mu \in \sigma_{|L}(i'))
\end{array}
\right)
\Rightarrow 
(a ~\multiAppendLeft~ \mu \in \sigma_{|L}(i))
\]
\end{lemma}

\begin{proof}
Let us consider $i$ and $i'$ in $\mathbb{I}(L)$ and $a$ in $\mathbb{A}(L)$ and $\mu \in \mathbb{M}(L)$.
Let us then suppose that $i \xrightarrow{a} i'$ and that $\mu \in \sigma_{|L}(i')$.
Let us then reason by induction on the cases that makes the hypothesis $i \xrightarrow{a} i'$ possible.
\begin{enumerate}
    \item when executing an atomic action, we have $i \in \mathbb{A}(L)$ and $i' = \varnothing$. Then $\sigma_{|L}(i) = \{i\}$ and $\sigma_{|L}(\varnothing) = \{\varepsilon_L\}$. The property $i ~\multiAppendLeft~ \varepsilon_L = i \in \sigma_{|L}(i)$ holds.
    \item when executing an action on the left of an alternative, we have $i$ of the form $alt(i_1,i_2)$, and $i' = i_1'$ such that $i_1 \xrightarrow{a} i_1'$. By construction of $\sigma_{|L}$, we have that $\mu \in \sigma_{|L}(i_1')$. By the induction hypothesis on the sub-interaction $i_1$, we have that $a ~\multiAppendLeft~ \mu \in \sigma_{|L}(i_1)$. Given that $\sigma_{|L}(i_1) \subset \sigma_{|L}(i)$, the property holds.
    \item executing actions on the right of an $alt$ can be treated similarly
    \item when executing an action on the left of a $par$, we have $i$ of the form $par(i_1,i_2)$, and $i' = par(i_1',i_2)$ such that $i_1 \xrightarrow{a} i_1'$. We have that $\mu \in \sigma_{|L}(par(i_1',i_2))$. By definition of $\sigma_{|L}$, we have that there exist $(\mu_1',\mu_2) \in \sigma_{|L}(i_1') \times \sigma_{|L}(i_2)$ s.t. $\mu \in (\mu_1' \multiInterleaving \mu_2)$. Therefore we have $i_1 \xrightarrow{a} i_1'$ and $\mu_1' \in \sigma_{|L}(i_1')$. Hence we can apply the induction hypothesis on sub-interaction $i_1$, which implies that $a ~\multiAppendLeft~ \mu_1' \in \sigma_{|L}(i_1)$. Given that $\sigma_{|L}(par(i_1,i_2))$ is the union of all the $(\mu_\alpha \multiInterleaving \mu_\beta)$ with $\mu_\alpha$ and $\mu_\beta$ multi-traces from $i_1$ and $i_2$, we have that $((a ~\multiAppendLeft~ \mu_1') \multiInterleaving \mu_2) \subset \sigma_{|L}(i)$. In particular, we know that $\mu \in (\mu_1' \multiInterleaving \mu_2)$, so, by definition of the $\multiInterleaving$ operator, we have that $a ~\multiAppendLeft~ \mu \in ((a ~\multiAppendLeft~\mu_1') \multiInterleaving \mu_2)$. Therefore the property holds.
    \item executing actions on the right of a $par$ can be treated similarly
    \item when executing an action on the left of a $seq$, we have $i$ of the form $seq(i_1,i_2)$, and $i' = seq(i_1',i_2)$ such that $i_1 \xrightarrow{a} i_1'$. We have that $\mu \in \sigma_{|L}(seq(i_1',i_2))$. By definition of $\sigma_{|L}$, we have that there exist $(\mu_1',\mu_2) \in \sigma_{|L}(i_1') \times \sigma_{|L}(i_2)$ s.t. $\mu \in (\mu_1' \multiSeq \mu_2)$. Therefore we have $i_1 \xrightarrow{a} i_1'$ and $\mu_1' \in \sigma_{|L}(i_1')$. Hence we can apply the induction hypothesis on sub-interaction $i_1$, which implies that $a ~\multiAppendLeft~ \mu_1' \in \sigma_{|L}(i_1)$. Given that $\sigma_{|L}(seq(i_1,i_2))$ is the union of all the $(\mu_\alpha \multiSeq \mu_\beta)$ with $\mu_\alpha$ and $\mu_\beta$ multi-traces from $i_1$ and $i_2$, we have that $((a ~\multiAppendLeft~ \mu_1') \multiSeq \mu_2) \subset \sigma_{|L}(i)$. In particular, we know that $\mu \in (\mu_1' \multiSeq \mu_2)$, so, by definition of the $\multiSeq$ operator, we have that $a ~\multiAppendLeft~ \mu \in ((a ~\multiAppendLeft~ \mu_1') \multiSeq \mu_2)$. Therefore the property holds.
    \item when executing an action on the right of a $seq$, we have $i$ of the form $seq(i_1,i_2)$, and $i' = i_2'$ such that $i_2 \xrightarrow{a} i_2'$ with the added hypothesis that $i_1\downarrow$. We have that $\mu \in \sigma_{|L}(i_2')$. Therefore we have $i_2 \xrightarrow{a} i_2'$ and $\mu \in \sigma_{|L}(i_2')$. Hence we can apply the induction hypothesis on sub-interaction $i_2$, which implies that $a ~\multiAppendLeft~ \mu \in \sigma_{|L}(i_2)$. Given that $\sigma_{|L}(seq(i_1,i_2))$ includes $\sigma_{|L}(i_2)$ when $i_1\downarrow$, and given that we know $i_1\downarrow$ to be true, the property holds.
    \item when executing an action underneath a $loop_S$, we have $i$ of the form $loop_S(i_1)$ and $i' = seq(i_1',loop_S(i_1))$ such that $i_1 \xrightarrow{a} i_1'$. We have that $\mu \in \sigma_{|L}(i')$. Therefore there exists $\mu_1 \in \sigma_{|L}(i_1')$ and $\mu_2 \in \sigma_{|L}(i)$ s.t. $\mu \in (\mu_1 \multiSeq \mu_2)$.
    \begin{itemize}
        \item We have $i_1 \xrightarrow{a} i_1'$ and $\mu_1 \in \sigma_{|L}(i_1')$. Hence we can apply the induction hypothesis on sub-interaction $i_1$, which implies that $a ~\multiAppendLeft~ \mu_1 \in \sigma_{|L}(i_1)$.
        \item As a result, given that $\mu_2 \in \sigma_{|L}(loop_S(i_1)) = \sigma_{|L}(i_1)^{\multiSeq *}$, and $a ~\multiAppendLeft~ \mu_1 \in \sigma_{|L}(i_1)$, we have, $((a ~\multiAppendLeft~ \mu_1) \multiSeq \mu_2) \subset \sigma_{|L}(i_1)^{\multiSeq *}$ i.e. $((a ~\multiAppendLeft~ \mu_1) \multiSeq \mu_2) \subset \sigma_{|L}(i)$
        \item Also, given that $\mu \in (\mu_1 \multiSeq \mu_2)$, we have immediately that $a ~\multiAppendLeft~ \mu \in ((a ~\multiAppendLeft~ \mu_1) \multiSeq \mu_2))$ because it is always possible to add actions from the left.
        \item Therefore $a ~\multiAppendLeft~ \mu \in \sigma_{|L}(i)$, so the property holds.
    \end{itemize}
    \item when executing an action underneath a $loop_P$, we have $i$ of the form $loop_P(i_1)$ and $i' = par(i_1',loop_P(i_1))$ such that $i_1 \xrightarrow{a} i_1'$. We have that $\mu \in \sigma_{|L}(i')$. Therefore there exists $\mu_1 \in \sigma_{|L}(i_1')$ and $\mu_2 \in \sigma_{|L}(i)$ s.t. $\mu \in (\mu_1 \multiInterleaving \mu_2)$.
    \begin{itemize}
        \item We have $i_1 \xrightarrow{a} i_1'$ and $\mu_1 \in \sigma_{|L}(i_1')$. Hence we can apply the induction hypothesis on sub-interaction $i_1$, which implies that $a ~\multiAppendLeft~ \mu_1 \in \sigma_{|L}(i_1)$.
        \item As a result, given that $\mu_2 \in \sigma_{|L}(loop_P(i_1)) = \sigma_{|L}(i_1)^{\multiInterleaving *}$, and $a ~\multiAppendLeft~ \mu_1 \in \sigma_{|L}(i_1)$, we have, $(a ~\multiAppendLeft~ \mu_1 \multiInterleaving \mu_2) \subset \sigma_{|L}(i_1)^{\multiInterleaving *}$ i.e. $(a ~\multiAppendLeft~ \mu_1 \multiInterleaving \mu_2) \subset \sigma_{|L}(i)$
        \item Also, given that $\mu \in (\mu_1 \multiInterleaving \mu_2)$, we have immediately that $a ~\multiAppendLeft~ \mu \in (a ~\multiAppendLeft~ \mu_1 \multiInterleaving \mu_2)$ because it is always possible to add actions from the left.
        \item Therefore $a ~\multiAppendLeft~ \mu \in \sigma_{|L}(i)$, so the property holds.
    \end{itemize}
\end{enumerate}
\qed 
\end{proof}

Thanks to the previous Lemma (Lem.\ref{lem:sem_de_execute1}) as well as the characterization from Lem.\ref{lem:sem_de_terminates}, we can conclude on the inclusion of the $\sigma_{o|L}$ semantics into the $\sigma_{|L}$ semantics. Indeed, those two Lemmas state that the $\sigma_{|L}$ semantics accepts the same two construction rules (that for the empty multi-trace $\varepsilon_L$ and that for non empty multi-traces) as those that define $\sigma_{o|L}$ inductively. As a result any multi-trace that might be accepted according to $\sigma_{o|L}$ must also be accepted according to $\sigma_{|L}$. However, it does not imply the reciprocate (i.e. whether or not $\sigma_{|L}$ is included in $\sigma_{o|L}$). Indeed, it may be so that, if it were formulated using construction rules, $\sigma_{|L}$ would also verify some other construction rules in addition to the aforementioned two, which would allow the acceptation of some more traces.

\begin{theorem}[Inclusion of $\sigma_{o|L}$ in $\sigma_{|L}$\label{th:sem_op_included_in_sem_de}]
For any interaction $i \in \mathbb{I}(L)$:
\[
\sigma_{o|L}(i) \subset \sigma_{|L}(i)
\]
\end{theorem}

\begin{proof}
Let us consider $i \in \mathbb{I}(L)$ and $\mu \in \sigma_{o|L}(i)$ and let us reason by induction on $\mu$.
\begin{itemize}
    \item If $\mu = \varepsilon_L$, then, as per the definition of $\sigma_{o|L}$, this means that $i\downarrow$. Then as per Lem.\ref{lem:sem_de_terminates}, this means that $\varepsilon_L \in \sigma_{|L}(i)$.
    \item If $\mu \neq \varepsilon_L$ then, by definition of $\sigma_{o|L}$, there exists $a \in \mathbb{A}(L)$, $\mu' \in \mathbb{M}(L)$ and $i' \in \mathbb{I}(L)$ s.t. $\mu = a ~\multiAppendLeft~ \mu'$, $i \xrightarrow{a} i'$ and $\mu' \in \sigma_{o|L}(i')$. By the induction hypothesis on $\mu'$, we have $(\mu' \in \sigma_{o|L}(i')) \Rightarrow (\mu' \in \sigma_{|L}(i'))$. As a result, we have $i \xrightarrow{a} i'$ and $\mu' \in \sigma_{|L}(i')$. We can therefore apply Lem.\ref{lem:sem_de_execute1} to conclude that $\mu = a ~\multiAppendLeft~ \mu' \in \sigma_{|L}(i)$. Hence the property holds.
\end{itemize}
\qed 
\end{proof}

\subsection{Right inclusion\label{anx:hyp_right}}

\begin{lemma}[Characterization (right side) of $\rightarrow$ w.r.t. $\sigma_{|L}$\label{lem:sem_de_execute2}]
For any multi-trace $\mu \in \mathbb{M}(L)$ and for any interaction $i \in \mathbb{I}(L)$:
\[
\left(
\begin{array}{l}
(\mu \in \sigma_{|L}(i))\\
\wedge (\mu \neq \varepsilon_L)
\end{array}
\right)
\Rightarrow
\left(
\exists~
\left\{
\begin{array}{l}
a \in \mathbb{A}(L),\\
\mu' \in \mathbb{M}(L),\\
\exists~i' \in \mathbb{I}(L)
\end{array}
\right\}
,~
\begin{array}{l}
(\mu = a ~\multiAppendLeft~ \mu'),\\
\wedge (i \xrightarrow{a} i')\\
\wedge (\mu' \in \sigma_{|L}(i'))
\end{array}
\right)
\]
\end{lemma}

\begin{proof}
Let us reason by induction on the term structure of $i$.
\begin{itemize}
    \item we cannot have $i = \varnothing$ because it contradicts $\mu \neq \varepsilon_L$
    \item if $i \in \mathbb{A}(L)$ then we have $\mu = i ~\multiAppendLeft~ \varepsilon_L$. We then have the existence of $i' = \varnothing$ which indeed satisfies that $i \xrightarrow{i} \varnothing$ and $\varepsilon_L \in \sigma_{|L}(\varnothing)$
    \item if $i$ is of the form $alt(i_1,i_2)$ then $\mu \in \sigma_{|L}(i)$ implies either $\mu \in \sigma_{|L}(i_1)$ or $\mu \in \sigma_{|L}(i_2)$. Let us suppose it is the first case (the second is identical). Then, we can apply the induction hypothesis on sub-interaction $i_1$, which reveals the existence of $a$, $\mu'$ and $i_1'$ such that $\mu = a ~\multiAppendLeft~ \mu'$, $i_1 \xrightarrow{a} i_1'$ and $\mu' \in \sigma_{|L}(i_1')$. By definition of the execution relation "$\rightarrow$", this implies that $alt(i_1,i_2) \xrightarrow{a} i_1'$. As a result, we have identified $i'=i_1'$ which satisfies the property.
    \item if $i$ is of the form $par(i_1,i_2)$ then $\mu \in \sigma_{|L}(i)$ implies the existence of multi-traces $\mu_1$ and $\mu_2$ such that $\mu \in \mu_1 \multiInterleaving \mu_2$, $\mu_1 \in \sigma_{|L}(i_1)$ and $\mu_2 \in \sigma_{|L}(i_2)$. Given that $\mu \neq \varepsilon_L$ we have either or both of $\mu_1 \neq \varepsilon_L$ and $\mu_2 \neq \varepsilon_L$. Let us suppose the first case (the other is similar). We then have by the induction hypothesis the existence of $a_1 \in \mathbb{A}(L)$, $\mu_1' \in \mathbb{M}(L)$ and $i_1' \in \mathbb{I}(L)$ such that $\mu_1 = a_1 ~\multiAppendLeft~ \mu_1'$, $i_1 \xrightarrow{a_1} i_1'$ and $\mu_1' \in \sigma_{|L}(i_1')$. Let us also suppose that $\mu$ can then be written as $\mu = a_1 ~\multiAppendLeft~ \mu'$ where $\mu' \in \mu_1' \multiInterleaving \mu_2$. If this is not the case then it means that there exists an action $a_2$ such that $\mu_2 = a_2 ~\multiAppendLeft~ \mu_2'$ and $\mu = a_2 ~\multiAppendLeft~ \mu''$ where $\mu'' \in \mu_1 \multiInterleaving \mu_2'$ and we can go back to the second case. In any case we now have $\mu = a_1 ~\multiAppendLeft~ \mu'$ and $par(i_1,i_2) \xrightarrow{a_1} par(i_1',i_2)$ with, by definition, $\mu' \in \sigma_{|L}(par(i_1',i_2))$. We therefore have identified $i' = par(i_1',i_2)$ and $\mu'$ which satisfy the property.
    \item if $i$ is of the form $seq(i_1,i_2)$ then there exist $\mu_1 \in \sigma_{|L}(i_1)$ and $\mu_2 \in \sigma_{|L}(i_2)$ such that $\mu \in (\mu_1 \multiSeq \mu_2)$. Then:
    \begin{itemize}
        \item if $\mu_1 \neq \varepsilon_L$, we can apply the induction hypothesis on sub-interaction $i_1$ s.t. we have the existence of $a_1$, $\mu_1'$ and $i_1'$ s.t. $\mu_1 = a_1 ~\multiAppendLeft~ \mu_1'$, $i_1 \xrightarrow{a_1} i_1'$ and $\mu_1' \in \sigma_{|L}(i_1')$. By definition of the execution relation "$\rightarrow$", this implies that $seq(i_1,i_2) \xrightarrow{a_1} seq(i_1',i_2)$. By definition of $\sigma_{|L}$, given that $\mu_1' \in \sigma_{|L}(i_1')$ and $\mu_2 \in \sigma_{|L}(i_2)$, we have $(\mu_1' \multiSeq \mu_2) \subset \sigma_{|L}(seq(i_1',i_2))$. Let us then denote by $\mu'$ the multi-trace such that $\mu = a_1 ~\multiAppendLeft~ \mu'$. Then, given that  $\mu' \in (\mu_1' \multiSeq \mu_2)$ this implies that $\mu' \in \sigma_{|L}(seq(i_1',i_2))$. We therefore have identified $i' = seq(i_1',i_2)$ and $\mu'$ which satisfy the property.
        \item if $\mu_1 = \varepsilon_L$ then, as per Lem.\ref{lem:sem_de_terminates}, we have $i_1\downarrow$. Also, because $\mu \neq \varepsilon_L$ we must have $\mu_2 \neq \varepsilon_L$. We can apply the induction hypothesis on sub-interaction $i_2$ s.t. we have the existence of $a_2$, $\mu_2'$ and $i_2'$ s.t. $\mu_2 = a_2 ~\multiAppendLeft~ \mu_2'$, $i_2 \xrightarrow{a_2} i_2'$ and $\mu_2' \in \sigma_{|L}(i_2')$. By definition of the execution relation "$\rightarrow$", and because the precondition $i_1\downarrow$ is verified, this implies that $seq(i_1,i_2) \xrightarrow{a} i_2'$. As a result, we have identified $i' = i_2'$ and $\mu' = \mu_2'$ which satisfy the property.
    \end{itemize}
    \item if $i$ is of the form $loop_S(i_1)$ then there exists $\mu_1 \in \sigma_{|L}(i_1)$ and $\mu_0 \in \sigma_{|L}(i)$ such that $\mu \in \mu_1 \multiSeq \mu_0$. Let us suppose that $\mu_1 \neq \varepsilon_L$ (otherwise we must have $\mu_0 \neq \varepsilon_L$ and we can be brought back to the same case). We can apply the induction hypothesis on sub-interaction $i_1$ s.t. we have the existence of $a_1$, $\mu_1'$ and $i_1'$ s.t. $\mu_1 = a_1 ~\multiAppendLeft~ \mu_1'$, $i_1 \xrightarrow{a_1} i_1'$ and $\mu_1' \in \sigma_{|L}(i_1')$. By definition of the execution relation "$\rightarrow$", this implies that $loop_S(i_1) \xrightarrow{a} seq(i_1',loop_S(i_1))$. Let us then denote by $\mu'$ the multi-trace s.t. $\mu = a_1 ~\multiAppendLeft~ \mu'$. Then, given that  $\mu' \in (\mu_1' \multiSeq \mu_0)$ this implies that $\mu' \in \sigma_{|L}(seq(i_1',loop_S(i_1)))$. We therefore have identified $i' = seq(i_1',loop_S(i_1))$ and $\mu'$ which satisfy the property.
    \item if $i$ is of the form $loop_P(i_1)$ then there exists $\mu_1 \in \sigma_{|L}(i_1)$ and $\mu_0 \in \sigma_{|L}(i)$ such that $\mu \in \mu_1 \multiInterleaving \mu_0$. Let us suppose that $\mu_1 \neq \varepsilon_L$ (otherwise we must have $\mu_0 \neq \varepsilon_L$ and we can be brought back to the same case). We can apply the induction hypothesis on sub-interaction $i_1$ s.t. we have the existence of $a_1$, $\mu_1'$ and $i_1'$ s.t. $\mu_1 = a_1 ~\multiAppendLeft~ \mu_1'$, $i_1 \xrightarrow{a_1} i_1'$ and $\mu_1' \in \sigma_{|L}(i_1')$. By definition of the execution relation "$\rightarrow$", this implies that $loop_P(i_1) \xrightarrow{a} par(i_1',loop_P(i_1))$. Let us then suppose the existence of $\mu'$ s.t. $\mu = a_1 ~\multiAppendLeft~ \mu'$ (otherwise a first action is taken from $\mu_0$ and we can go back to the same case). Then, given that  $\mu' \in (\mu_1' \multiInterleaving \mu_0)$ this implies that $\mu' \in \sigma_{|L}(par(i_1',loop_P(i_1)))$. We therefore have identified $i' = par(i_1',loop_P(i_1))$ and $\mu'$ which satisfy the property.
\end{itemize}
\qed 
\end{proof}

Thanks to the Lem.\ref{lem:sem_de_execute2} as well as the characterization from Lem.\ref{lem:sem_de_terminates} we can conclude on the inclusion of the $\sigma_{|L}$ semantics into the $\sigma_{o|L}$ semantics.

\begin{theorem}[Inclusion of $\sigma_{|L}$ in $\sigma_{o|L}$\label{th:sem_de_included_in_sem_op}]
For any interaction $i \in \mathbb{I}(L)$:
\[
\sigma_{o|L}(i) \supset \sigma_{|L}(i)
\]
\end{theorem}

\begin{proof}
Let us consider $i \in \mathbb{I}(L)$ and $\mu \in \sigma_{|L}(i)$ and let us reason by induction on the size of $\mu$.
\begin{itemize}
    \item If $\mu = \varepsilon_L$, the fact that $\mu = \varepsilon_L \in \sigma_{|L}(i)$ implies, as per Lem.\ref{lem:sem_de_terminates}, that $i\downarrow$. Then, by definition of $\sigma_{o|L}$, this means that $\varepsilon_L \in \sigma_{o|L}(i)$.
    \item If $\mu \neq \varepsilon$ then, as per Lem.\ref{lem:sem_de_execute2} this implies the existence of $a$, $\mu'$ and $i'$ s.t. $\mu = a ~\multiAppendLeft~ \mu'$, $i \xrightarrow{a} i'$ and $\mu' \in \sigma_{|L}(i')$. Because $\mu'$ is of a smaller size than $\mu$ (minus one), we can apply the induction hypothesis so that we have $\mu' \in \sigma_{o|L}(i')$. Given that we also have $i \xrightarrow{a} i'$, we have by definition $\mu = a ~\multiAppendLeft~ \mu' \in \sigma_{o|L}(i)$.
\end{itemize}
\qed 
\end{proof}

We have finally proven both inclusion and we conclude with Th.\ref{th:sem_equivalence_de_op} that the operational semantics $\sigma_{o|L}$ that we have defined in Def.\ref{def:interaction_operational_semantics} is indeed equivalent to the denotational-style semantics $\sigma_{|L}$ from Def.\ref{def:algebraic_multi_trace_semantics}.

\begin{theorem}[Equivalence of the $\sigma_{|L}$ and $\sigma_{o|L}$ semantics\label{th:sem_equivalence_de_op}]
For any interaction $i \in \mathbb{I}(L)$:
\[
\sigma_{o|L}(i) = \sigma_{|L}(i)
\]
\end{theorem}

\begin{proof}
Immediately implies by Th.\ref{th:sem_op_included_in_sem_de} and Th.\ref{th:sem_de_included_in_sem_op}. 
\qed 
\end{proof}

\clearpage

\section{Details on the 3SAT experiments\label{anx:exp_3sat}}

We provide an experimental validation of the implementation in HIBOU of the algorithm from Section \ref{ssec:algo} (which includes additional optimizations). It consists in testing the reliability of the algorithm (i.e. whether it returns a $Pass$ or a $Fail$) on a set of problems of which we know the answer. Those problems are obtained via reduction from benchmarks of 3SAT problems using the polynomial reduction presented in Property~\ref{prop:multipref_nphard_complexity} in Section \ref{ssec:complexity}.

The polynomial reduction from 3\,SAT to multi-trace analysis was implemented as a Python script which translates ".cnf" files in the DIMACS \cite{dimacs_format} format into the entry language of HIBOU.
Those experiments serve as a validation of the tool i.e. that we can indeed correctly differentiate between satisfiable and unsatisfiable problems. For comparing the results of HIBOU, we used the Varisat \cite{varisat_rs} solver.
We then verify that both tools obtain the same satisfiability result on all problems and compare the time required to obtain those results, keeping the median value of $5$ tries to smooth the data. All the related material and code for reproducing those results are available in \cite{hibou_3sat_bench}.

As input data we have used 3 sets of problems: two custom benchmarks with randomly generated problems
and the UF20 benchmark \cite{sat_benchmarks}.

Fig.\ref{fig:anx_hibou_varisat_3sat} provide details on each benchmark with, on the top left, information about the input problems (numbers of variables, clauses, instances), on the bottom left statistical information about the time required for the analysis using each tool, and, on the right a corresponding scatter plot. In the plot, each point corresponds to a given 3-SAT problem, with its position corresponding to the time required to solve it (by Varisat on the $x$ axis and HIBOU on the $y$ axis). Points in red are unsatisfiable problems while those in blue are satisfiable.

\begin{figure}[h]
    \centering
    \begin{subfigure}[t]{.975\linewidth}
        \centering
        \begin{minipage}{.49\linewidth}
\centering
{
\scriptsize  

\begin{tabular}{|l|r|}
\hline 
\# variables & 3-10 \\
\hline 
\# clauses & 4-50 \\
\hline 
\# instances & 663 \\
\hline 
\# SAT & 376 \\
\hline 
\# UNSAT & 287 \\
\hline 
\end{tabular}

~\\~\\

\begin{tabular}{|c|l|l|}
\cline{2-3}
\multicolumn{1}{c|}{} & \multicolumn{1}{c|}{varisat} & \multicolumn{1}{c|}{hibou}\\
\hline 
min & 0.01699 & 0.0002379 \\
\hline 
q1 & 0.01792 & 0.0012984 \\
\hline 
Mdn & 0.01806 & 0.0027920 \\
\hline 
M & 0.01833 & 0.0043448 \\
\hline 
q3 & 0.01848 & 0.0053158 \\
\hline 
max & 0.02892 & 0.0267174 \\
\hline 
$\sigma$ & 0.001017846 & 0.004637261 \\
\hline 
\end{tabular}
}       
\end{minipage} 
        \begin{minipage}{.49\linewidth}
            \centering
            \includegraphics[width=\textwidth]{images/experiments/3sat/Custom_small.png}
        \end{minipage}
        \caption{Input problems and output results for 'small' custom benchmark}
    \end{subfigure}
    \begin{subfigure}[t]{.975\linewidth}
        \centering
        \begin{minipage}{.49\linewidth}
\centering
{
\scriptsize  

\begin{tabular}{|l|r|}
\hline 
\# variables & 20-27 \\
\hline 
\# clauses & 40-100 \\
\hline 
\# instances & 790 \\
\hline 
\# SAT & 488 \\
\hline 
\# UNSAT & 302 \\
\hline 
\end{tabular}

~\\~\\

\begin{tabular}{|c|l|l|}
\cline{2-3}
\multicolumn{1}{c|}{} & \multicolumn{1}{c|}{varisat} & \multicolumn{1}{c|}{hibou}\\
\hline 
min & 0.01559 & 0.00246 \\
\hline 
q1 & 0.01808 & 0.03266 \\
\hline 
Mdn & 0.01895 & 0.49785 \\
\hline 
M & 0.01927 & 2.30209 \\
\hline 
q3 & 0.01995 & 1.93719 \\
\hline 
max & 0.02838 & 47.00918 \\
\hline 
$\sigma$ & 0.001813495 & 5.123452 \\
\hline 
\end{tabular}
}       
\end{minipage}
        \begin{minipage}{.49\linewidth}
            \centering
            \includegraphics[width=\textwidth]{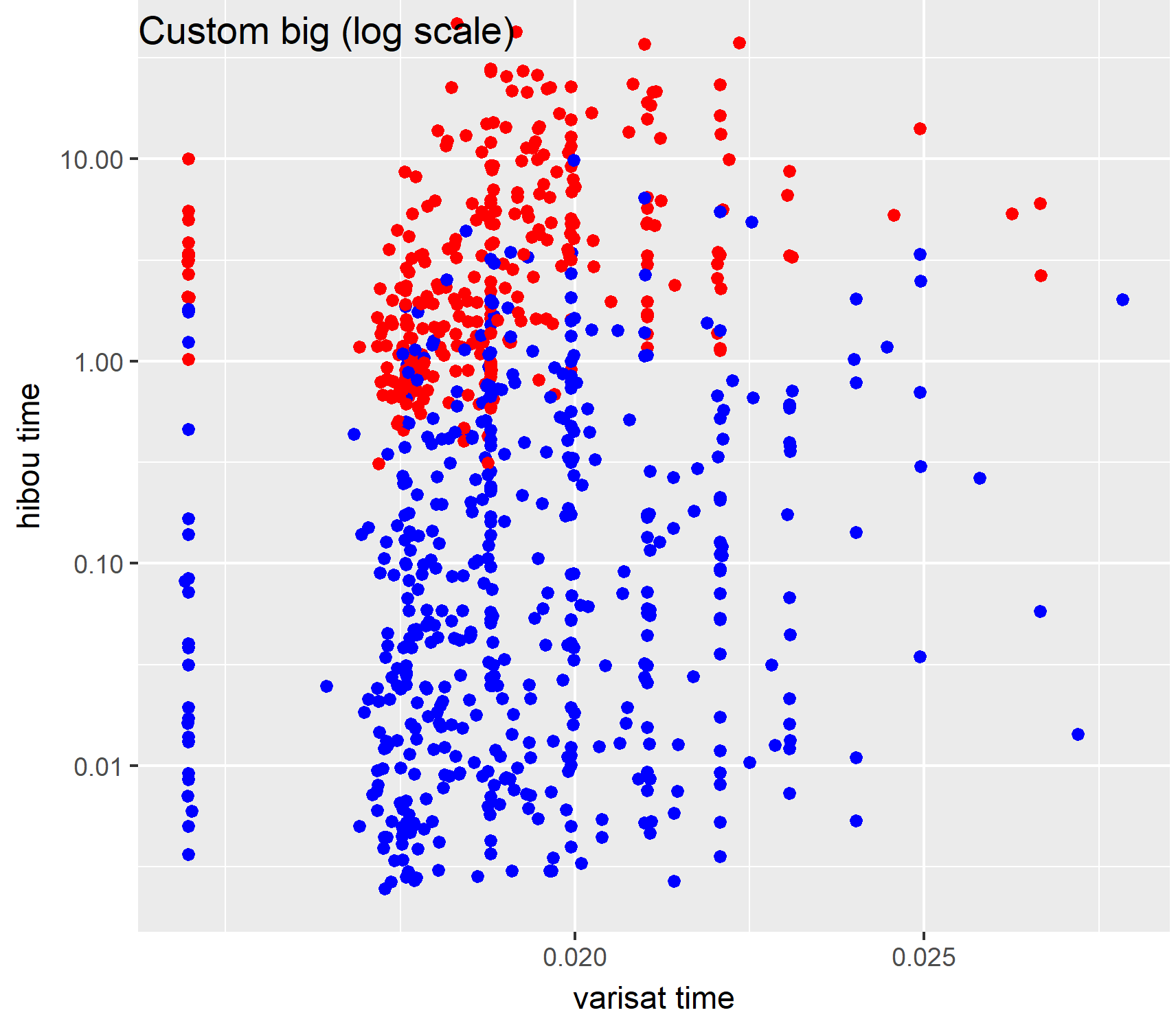}
        \end{minipage}
        \caption{Input problems and output results for 'big' custom benchmark}
    \end{subfigure}
    \begin{subfigure}[t]{.975\linewidth}
        \centering
        \begin{minipage}{.49\linewidth}
\centering
{
\scriptsize  

\begin{tabular}{|l|r|}
\hline 
\# variables & 20 \\
\hline 
\# clauses & 91 \\
\hline 
\# instances & 1000 \\
\hline 
\# SAT & 1000 \\
\hline 
\# UNSAT & 0 \\
\hline 
\end{tabular}

~\\~\\

\begin{tabular}{|c|l|l|}
\cline{2-3}
\multicolumn{1}{c|}{} & \multicolumn{1}{c|}{varisat} & \multicolumn{1}{c|}{hibou}\\
\hline 
min & 0.01559 & 0.007638 \\
\hline 
q1 & 0.01667 & 0.091421 \\
\hline 
Mdn & 0.01833 & 0.229745 \\
\hline 
M & 0.01847 & 0.313901 \\
\hline 
q3 & 0.01929 & 0.462385 \\
\hline 
max & 0.03989 & 1.666777 \\
\hline 
$\sigma$ & 0.00255181 & 0.2865485 \\
\hline 
\end{tabular}
}       
\end{minipage}
        \begin{minipage}{.49\linewidth}
            \centering
            \includegraphics[width=\textwidth]{images/experiments/3sat/UF20.png}
        \end{minipage}
        \caption{Input problems and output results for UF-20 benchmark}
    \end{subfigure}
    \caption{Experiments on 3SAT benchmarks (times in seconds)\label{fig:anx_hibou_varisat_3sat}}
    \vspace{-.5cm}
\end{figure}

\clearpage

\section{Details on the use cases experiments\label{anx:exp_usecases}}

We consider four use case interactions:
\begin{enumerate}
    \item A simple interaction protocol describing the purchase of a book. It is represented on Fig.\ref{fig:anx_coping}. This protocol corresponds to the example provided in \cite{coping_with_bad_agent_interaction_protocols_when_monitoring_partially_observable_multiagent_systems_AnconaFFM18} which we have adapted to be represented as an interaction in our language. We have also added a loop so that is can express arbitrarily long behaviors. The loop being a parallel loop $loop_P$, several instances of the repeatable behavior can be executed at the same time, potentially creating numerous possible interleavings of actions.
    \item A usecase on a system for querying sensor data. It is represented on Fig.\ref{fig:anx_sensor}. This usecase is inspired by \cite{a_dynamic_and_context_aware_semantic_mediation_service_for_discovering_and_fusion_of_heterogeneous_sensor_data}.
    \item A modelisaton of the Alternating Bit Protocol. It is represented on Fig.\ref{fig:anx_abp}. This example is inspired by that found in \cite{high_level_message_sequence_charts}.
    \item A usecase on a system for uploading data to a server. It is represented on Fig.\ref{fig:anx_network}. This example is inspired by that found in \cite{comprehensive_multiparty_session_types}. We have also added a loop for having arbitrarily long behaviors.
\end{enumerate}

For each example, we generate a number of accepted multi-traces using a trace generation feature of HIBOU. Because those interactions contain loops, the exploration of the model's semantics for trace generation must be stopped by a certain criterion. For each example, the exploration criterion is given on the left of the example's corresponding Figure.

Then, for each accepted multi-trace, we select a number of prefixes according to a certain selection criterion. For each example, the prefix selection criterion is given on the left of the example's corresponding Figure.

Then, for each prefix, we generate a number of mutants which can be of three kinds:
\begin{itemize}
    \item "noise" mutants consists in inserting additional random actions to the multi-trace (on the correct local component according to the action's lifeline of occurence)
    \item "swap action" mutants consists in swapping the positions of two actions within the same local component of the multi-trace
    \item "swap component" mutants are created by merging two distinct multi-prefixes, taking some local components from each one. Those mutants are peculiar because, by construction, all their local components are correct locally, but the global scenario which they describe may not necessarily be correct. 
\end{itemize}

\clearpage

\begin{figure}[t]
    \centering
    \begin{subfigure}[t]{.975\linewidth}
        \centering
        \includegraphics[scale=.25]{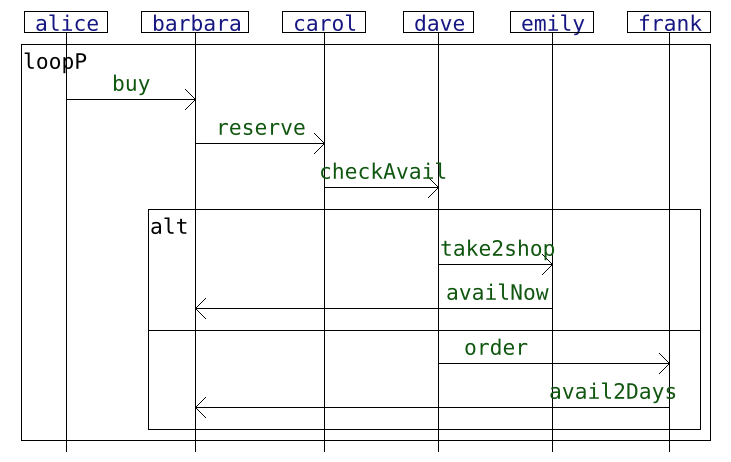}
        \caption{Diagram representation}
    \end{subfigure}

\vspace*{.5cm}

    \begin{subfigure}[t]{.975\linewidth}
        \centering 
            
\begin{tabular}{|c|l|}
\hline 
\makecell{\footnotesize Exploration\\\footnotesize criteria}
&
\makecell{\scriptsize loop $\leq 2$\\\scriptsize exhaustive}
\\
\hline 
\makecell{\footnotesize Prefix\\\footnotesize selection}
&
\makecell{\scriptsize $5$ random prefixes\\\scriptsize per trace}
\\
\hline 
\makecell{\footnotesize Mutant\\\footnotesize selection}
&
\makecell{\scriptsize $1$ mutant\\\scriptsize of each kind\\\scriptsize per prefix}
\\
\hline 
\end{tabular}

        \caption{Selection criteria}
    \end{subfigure}

\vspace*{.5cm}

    \begin{subfigure}[t]{.975\linewidth}
        \centering
        \includegraphics[width=.7\textwidth]{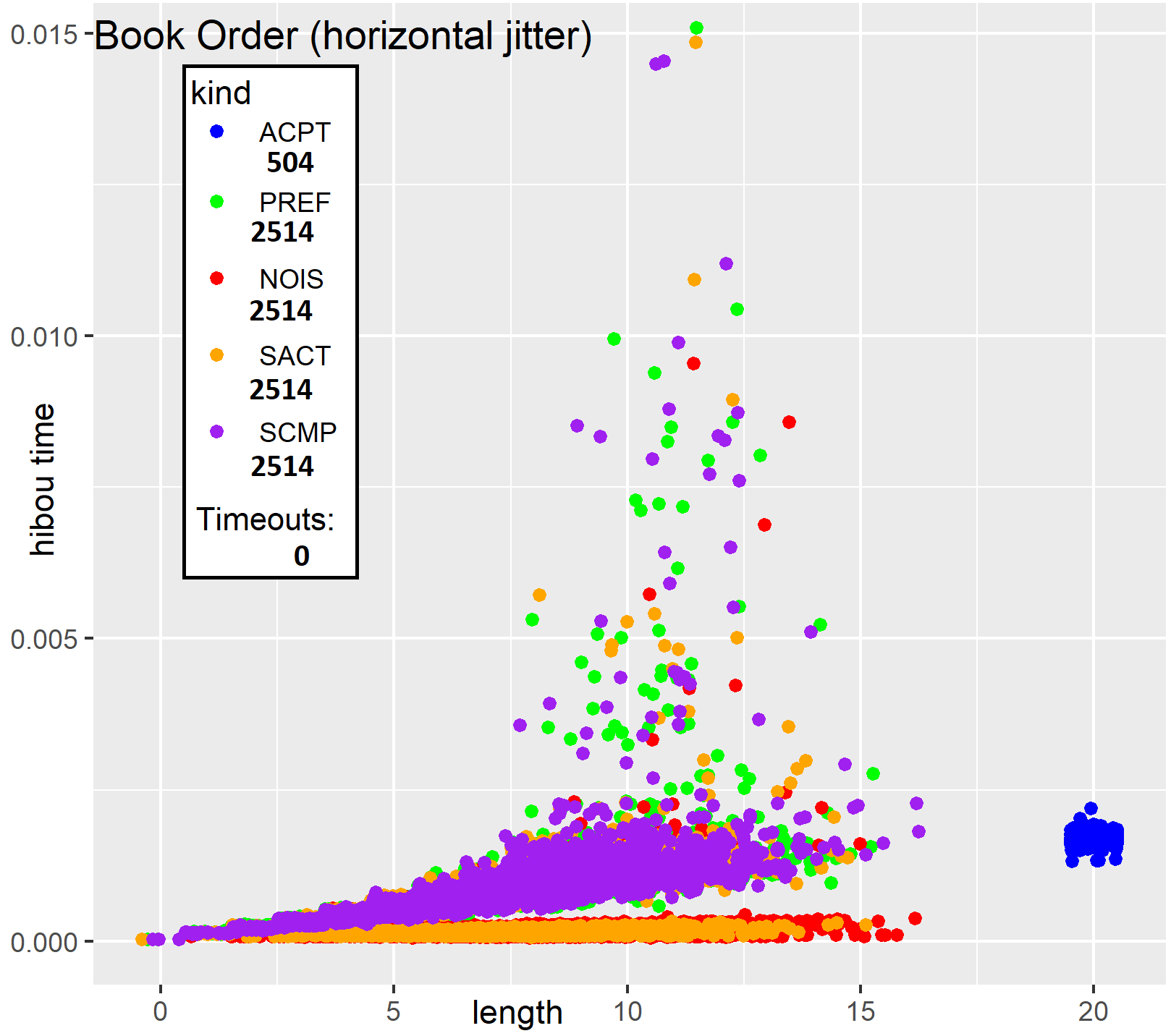}
        \caption{Experimental data}
    \end{subfigure}

\vspace*{.5cm}

    \begin{subfigure}[t]{.975\linewidth}
        \centering

\begin{tabular}{|l|l|l|l|l|}
\hline
min & q1 & M & q3 & max\\
\hline
0.0000268 & 0.0002316 & 0.0007333 & 0.0010154 & 0.0151036\\
\hline
\end{tabular}

\vspace*{.1cm}

\begin{tabular}{|l|l|l|l|l|l|l|}
\hline
Mdn & $\sigma$\\
\hline
0.0006283 & 0.0007559154\\
\hline
\end{tabular}

        \caption{Statistics}
    \end{subfigure}

    \caption{Book Order example from \cite{coping_with_bad_agent_interaction_protocols_when_monitoring_partially_observable_multiagent_systems_AnconaFFM18}}
    \label{fig:anx_coping}
\end{figure}

\clearpage

\begin{figure}[t]
    \centering
    \begin{subfigure}[t]{.975\linewidth}
        \centering
        \includegraphics[scale=.25]{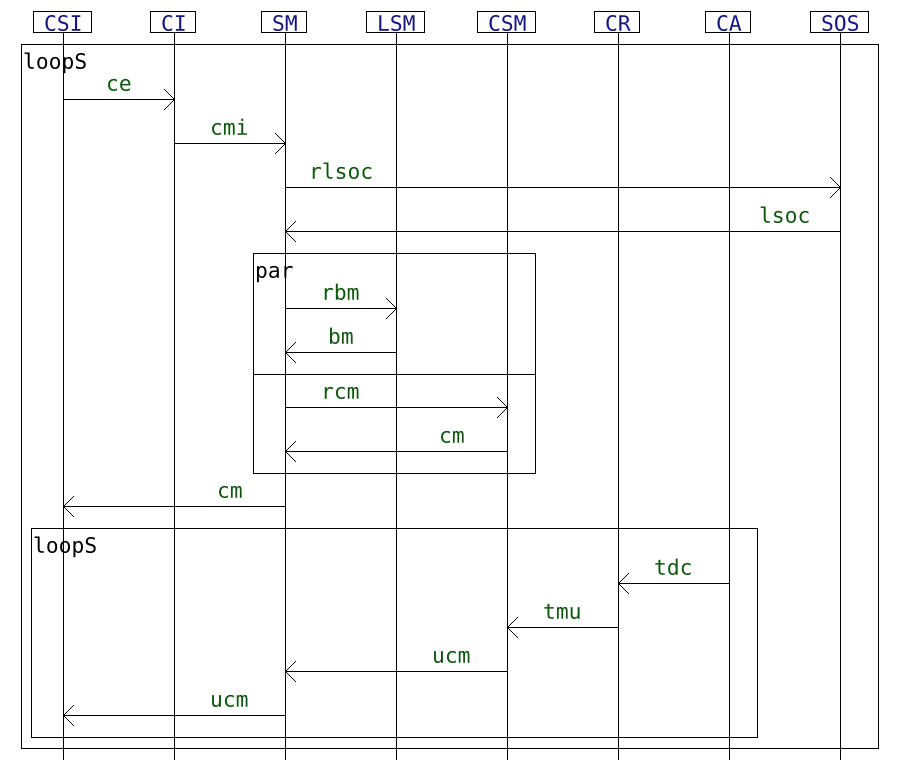}
        \caption{Diagram representation}
    \end{subfigure}

\vspace*{.5cm}

    \begin{subfigure}[t]{.975\linewidth}
        \centering 
            
\begin{tabular}{|l|l|l|}
\hline 
\makecell{\footnotesize Exploration\\\footnotesize criteria}
&
\makecell{\footnotesize Prefix\\\footnotesize selection}
&
\makecell{\footnotesize Mutant\\\footnotesize selection}
\\
\hline
\makecell{\scriptsize loop $\leq 50$\\\scriptsize partial \& random\\\scriptsize node $\leq 10~000$}
&
\makecell{\scriptsize $3$ random prefixes\\\scriptsize per trace}
&
\makecell{\scriptsize $1$ mutant\\\scriptsize of each kind\\\scriptsize per prefix}
\\
\hline 
\end{tabular}

        \caption{Selection criteria}
    \end{subfigure}

\vspace*{.5cm}

    \begin{subfigure}[t]{.975\linewidth}
        \centering
        \includegraphics[width=.75\textwidth]{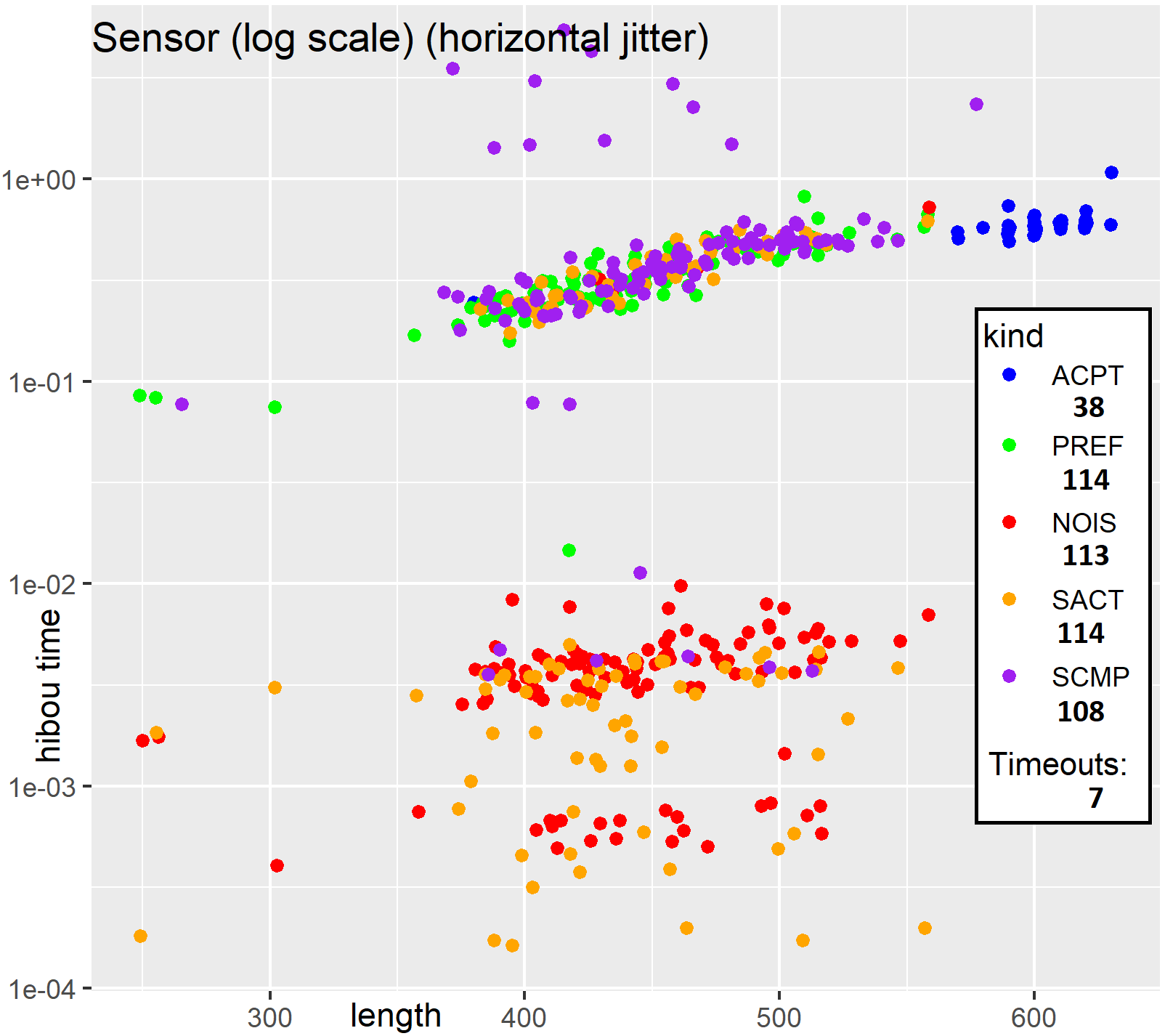}
        \caption{Experimental data}
    \end{subfigure}

\vspace*{.5cm}

    \begin{subfigure}[t]{.975\linewidth}
        \centering

\begin{tabular}{|l|l|l|l|l|}
\hline
min & q1 & M & q3 & max\\
\hline
0.000163 & 0.004035 & 0.297581 & 0.447480 & 5.480777 \\
\hline
\end{tabular}

\vspace*{.1cm}

\begin{tabular}{|l|l|l|l|l|l|l|}
\hline
Mdn & $\sigma$\\
\hline
0.256227 & 0.4673756\\
\hline
\end{tabular}

        \caption{Statistics}
    \end{subfigure}

    \caption{Sensor example from \cite{a_dynamic_and_context_aware_semantic_mediation_service_for_discovering_and_fusion_of_heterogeneous_sensor_data}}
    \label{fig:anx_sensor}
\end{figure}

\clearpage

\begin{figure}[t]
    \centering
    \begin{minipage}{.6\linewidth}
        \centering
        \begin{subfigure}[t]{.975\linewidth}
            \centering 
            
\begin{tabular}{|c|l|}
\hline 
\makecell{\footnotesize Exploration\\\footnotesize criteria}
&
\makecell{\scriptsize loop $\leq 15$\\\scriptsize partial \& random\\\scriptsize node $\leq 35~000$}
\\
\hline 
\makecell{\footnotesize Prefix\\\footnotesize selection}
&
\makecell{\scriptsize $5$ random prefixes\\\scriptsize per trace}
\\
\hline 
\makecell{\footnotesize Mutant\\\footnotesize selection}
&
\makecell{\scriptsize $1$ mutant\\\scriptsize of each kind\\\scriptsize per prefix}
\\
\hline 
\end{tabular}

            \caption{Selection criteria.}
        \end{subfigure}

\vspace*{.5cm}

        \begin{subfigure}[t]{.975\linewidth}
            \centering
            \includegraphics[width=\textwidth]{images/experiments/plots/abp_rw.png}
            \caption{Experimental data}
        \end{subfigure}

\vspace*{.5cm}

        \begin{subfigure}[t]{.975\linewidth}
            \centering

\begin{tabular}{|l|l|l|l|l|}
\hline
min & q1 & M & q3 & max\\
\hline
0.000119 & 0.001764 & 0.007855 & 0.006350 & 3.206065\\
\hline
\end{tabular}

\vspace*{.1cm}

\begin{tabular}{|l|l|l|l|l|l|l|}
\hline
Mdn & $\sigma$\\
\hline
0.003609 & 0.04388982\\
\hline
\end{tabular}

            \caption{Statistics}
        \end{subfigure}
    \end{minipage}
    \begin{minipage}{.39\linewidth}
        \centering
        \begin{subfigure}[t]{.975\linewidth}
            \centering
            \includegraphics[scale=.25]{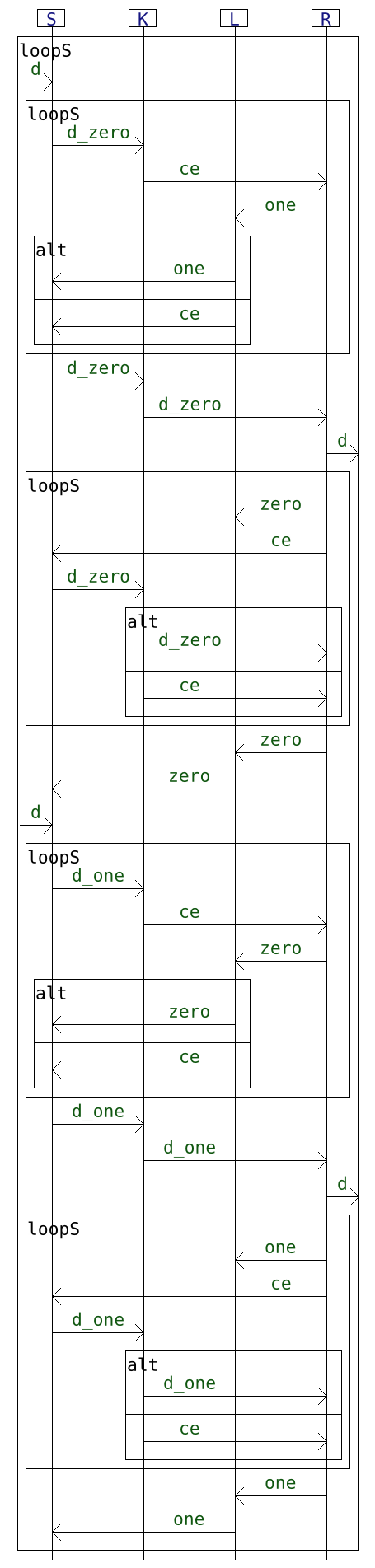}
            \caption{Diagram representation}
        \end{subfigure}
    \end{minipage}
    \caption{Alternating Bit Protocol example adapted from \cite{high_level_message_sequence_charts}.}
    \label{fig:anx_abp}
\end{figure}

\clearpage

\begin{figure}[t]
    \centering
    \begin{subfigure}[t]{.975\linewidth}
        \centering
        \includegraphics[scale=.25]{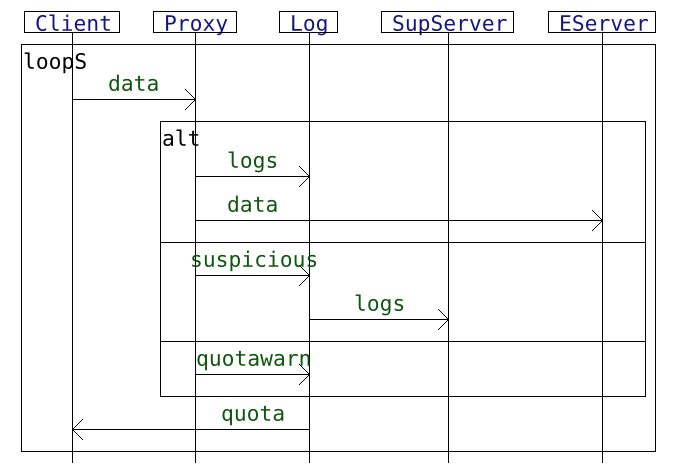}
        \caption{Diagram representation}
    \end{subfigure}

\vspace*{.5cm}

    \begin{subfigure}[t]{.975\linewidth}
        \centering 
            
\begin{tabular}{|c|l|}
\hline 
\makecell{\footnotesize Exploration\\\footnotesize criteria}
&
\makecell{\scriptsize loop $\leq 10$\\\scriptsize partial \& random\\\scriptsize node $\leq 20~000$}
\\
\hline 
\makecell{\footnotesize Prefix\\\footnotesize selection}
&
\makecell{\scriptsize $5$ random prefixes\\\scriptsize per trace}
\\
\hline 
\makecell{\footnotesize Mutant\\\footnotesize selection}
&
\makecell{\scriptsize $1$ mutant\\\scriptsize of each kind\\\scriptsize per prefix}
\\
\hline 
\end{tabular}

        \caption{Selection criteria}
    \end{subfigure}

\vspace*{.5cm}

    \begin{subfigure}[t]{.975\linewidth}
        \centering
        \includegraphics[width=.75\textwidth]{images/experiments/plots/network_rw.png}
        \caption{Experimental data}
    \end{subfigure}

\vspace*{.5cm}

    \begin{subfigure}[t]{.975\linewidth}
        \centering

\begin{tabular}{|l|l|l|l|l|}
\hline
min & q1 & M & q3 & max\\
\hline
0.0000593 & 0.0005464 & 0.0033469 & 0.0055067 & 0.0212768 \\
\hline
\end{tabular}

\vspace*{.1cm}

\begin{tabular}{|l|l|l|l|l|l|l|}
\hline
Mdn & $\sigma$\\
\hline
0.0035957 & 0.002896387\\
\hline
\end{tabular}

        \caption{Statistics}
    \end{subfigure}

    \caption{Network example from \cite{comprehensive_multiparty_session_types}}
    \label{fig:anx_network}
\end{figure}

\end{document}